\documentclass{ecai} 
\usepackage{latexsym}
\usepackage{amssymb}
\usepackage{amsmath}
\usepackage{amsthm}
\usepackage{booktabs}
\usepackage{enumitem}
\usepackage{graphicx}
\usepackage{color}

\usepackage{times}
\usepackage{soul}
\usepackage{url}
\usepackage[switch]{lineno}
\usepackage[linesnumbered,ruled,vlined]{algorithm2e}
\usepackage{multirow}
\usepackage{mathrsfs}
\usepackage{float}
\usepackage{setspace}
\usepackage{subfigure}
\usepackage{xcolor}

\newtheorem{theorem}{Theorem}
\newtheorem{lemma}[theorem]{Lemma}

\newtheorem{definition}{Definition}

\newtheorem{bounding}{Bounding Rule}
\newcommand{\tabincell}[2]{\begin{tabular}{@{}#1@{}}#2\end{tabular}}
\newcommand{\BibTeX}{B\kern-.05em{\sc i\kern-.025em b}\kern-.08em\TeX}

\begin{document}

%%%%%%%%%%%%%%%%%%%%%%%%%%%%%%%%%%%%%%%%%%%%%%%%%%%%%%%%%%%%%%%%%%%%%%%%

\begin{frontmatter}

%%% Use this command to specify your submission number.
%%% In doubleblind mode, it will be printed on the first page.

\paperid{123} 

%%% Use this command to specify the title of your paper.

\title{A Faster Branching Algorithm for the Maximum $k$-Defective Clique Problem}

%%% Use this combinations of commands to specify all authors of your 
%%% paper. Use \fnms{} and \snm{} to indicate everyone's first names 
%%% and surname. This will help the publisher with indexing the 
%%% proceedings. Please use a reasonable approximation in case your 
%%% name does not neatly split into "first names" and "surname".
%%% Specifying your ORCID digital identifier is optional. 
%%% Use the \thanks{} command to indicate one or more corresponding 
%%% authors and their email address(es). If so desired, you can specify
%%% author contributions using the \footnote{} command.

%\author[A]
%{\fnms{First}~\snm{Author}\orcid{....-....-....-....}\thanks{Corresponding Author. Email: somename@university.edu.}\footnote{Equal contribution.}}
\author[A]
{\fnms{Chunyu}~\snm{Luo}}

%{0000-0002-9023-4374}
\author[A]{\fnms{Yi}~\snm{Zhou}\orcid{0000-0002-9023-4374}\thanks{Corresponding Author. Email: zhou.yi@uestc.edu.cn}}

\author[B]{\fnms{Zhengren}~\snm{Wang}} 

\author[A]{\fnms{Mingyu}~\snm{Xiao}\orcid{0000-0002-1012-2373}}
%\author[B]{\fnms{Yi}~\snm{Zhou}\orcid{0000-0002-9023-4374}\footnotemark}
%\author[B,C]{\fnms{Third}~\snm{Author}\orcid{....-....-....-....}} 

\address[A]{University of Electronic Science and Technology of China, China}
\address[B]{Peking University, China}
%\address[C]{Short Alternate Affiliation of Third Author}
\if 0
\author{
    % Author Name
    % \affiliations
    % Affiliation
    % \emails
    % email@example.com
    Paper ID: 749
}
\fi
%%% Use this environment to include an abstract of your paper.

\begin{abstract}
A $k$-defective clique of an undirected graph $G$ is a subset of its vertices that induces a nearly complete graph with a maximum of $k$ missing edges.
        The maximum $k$-defective clique problem, which asks for the largest $k$-defective clique from the given graph, is important in many applications, such as social and biological network analysis.  
        In the paper, we propose a new branching algorithm that takes advantage of the structural properties of the $k$-defective clique and uses the efficient maximum clique algorithm as a subroutine.
        As a result, the algorithm has a better asymptotic running time than the existing ones.
        We also investigate upper-bounding techniques and propose a new upper bound utilizing the \textit{conflict relationship} between vertex pairs. 
        Because conflict relationship is common in many graph problems, we believe that this technique can be potentially generalized.        
        Finally, experiments show that our algorithm outperforms state-of-the-art solvers on a wide range of open benchmarks.
\end{abstract}

\end{frontmatter}

%%%%%%%%%%%%%%%%%%%%%%%%%%%%%%%%%%%%%%%%%%%%%%%%%%%%%%%%%%%%%%%%%%%%%%%%

\section{Introduction}
\if 0
The European Conference of Artificial Intelligence (ECAI) is the leading 
discipline-wide conference on AI in Europe. Its history goes back all 
the way to the Summer Conference on Artificial Intelligence and 
Simulation of Behaviour held in July 1974 in Brighton. Nowadays, ECAI is 
organised annually under the auspices of the European Association for 
Artificial Intelligence (EurAI, see Figure~\ref{fig:eurai}).

\begin{figure}[h]
\centering
\includegraphics[width=2.5cm]{eurai}
\caption{Logo of the European Association for Artificial Intelligence.}
\label{fig:eurai}
\end{figure}

Your paper should be typeset in \LaTeX, using the ECAI class file 
provided (\texttt{ecai.cls}). Please do not modify the class file or 
any of the layout parameters.

For instructions on how to submit your work to ECAI and on matters such 
as page limits or referring to supplementary material, please consult 
the Call for Papers of the next edition of the conference. Keep in mind
that you must use the \texttt{doubleblind} option for submission. 
\fi

Finding cohesive subgraphs is a fundamental task in many real-world network applications, such as community detection in social networks \citep{bedi2016community,azaouzi2019community,du2007community}, protein complex mining in biological networks \citep{yu2006predicting,harley2001uniform}, and statistical analysis in financial networks \citep{boginski2006mining,boginski2005statistical}. A recent study in ~\cite{matsugu2023uncovering} also reveals a potential application of finding cohesive graphs in the design of Graph Neural Network ~\cite{matsugu2023uncovering}.
The clique, which requires an edge for every pair of vertices, is a basic model for representing cohesive subgraphs. Additionally, the study of algorithms for the maximum clique problem, that is, finding the maximum clique from a given graph, has made significant progress both theoretically and practically in recent years. 
For example, the best-known time complexity for the exact computation of the maximum clique is $O^*(1.20^n)$ \citep{xiao2017exact}; Practical algorithms such as MC-BRB \citep{chang2019efficient} and MoMC \citep{li2017minimization} can solve massive sparse or hard dense graphs in seconds.
However, the clique model itself is not a preferred choice for real-world applications.
Due to the existence of noises and faults in many scenarios, the constraint that requires all possible relations to exist in the community may be too restrictive \citep{conte2018d2k}. 
Therefore, various clique relaxations have been formulated in the literature, such as $k$-plex \citep{wang2022listing,wang2023gap}, $k$-clique \citep{behar2018finding} and $k$-defective clique \citep{yu2006predicting,gao2022exact,chang2023efficient,jin2024kd}. 

This paper focuses on the $k$-defective clique model, a subset of vertices that allows the induced subgraph missing at most $k$ edges to be a complete subgraph. 
When $k$ is $0$, the $k$-defective clique is a clique.
%In this sense, the $k$-defective clique problem is NP-hard.
For any $k\ge 0$, the $k$-defective clique problem is NP-hard \citep{yannakakis1978node} and W[1]-hard with respect to its solution size \citep{khot2002parameterized}.
Furthermore, the problem cannot be approximated with a factor better than $O(n^\epsilon)$ for any $\epsilon>0$ currently \citep{haastad1999clique}.

Although the maximum $k$-defective problem is difficult in theory, there have been several practical algorithms in the last decade, such as \citep{trukhanov2013algorithms,gschwind2018maximum,chen2021computing,gao2022exact,chang2023efficient,jin2024kd}.
All of these algorithms follow the branching algorithm scheme (also known as branch-and-bound, tree search, or backtracking search). Furthermore, \citep{dai2023maximal} studied the enumeration of maximal $k$-defective cliques, which is also based on the branching algorithm. %\wzr{dai2023 is Unnecessary}

From a complexity perspective, under the assumption that $k$ is constant, branching algorithms have worst-case time complexity in the form of $O^*(\gamma_k^n)$ \footnote{The $O^*$ notation omits the polynomial factors.} where $n$ is the vertex number of the input graph and $\gamma_k$ is a constant related to $k$ but strictly smaller than 2 \citep{chen2021computing, chang2023efficient}.
Currently, the cutting-edge $\gamma_k$ is obtained by \citep{chang2023efficient}, for example $\gamma_1=1.84$ and $\gamma_2=1.93$.  
% \wzr{$\gamma_0$ should be emphasized}

% From a practical perspective, all these exact algorithms design pruning rules to reduce the number of branches. 
From a practical perspective, the efficiency of branching algorithms is closely related to the tightness of upper bounds to prune unpromising branches. (That is why branching algorithms are also named branch-and-bound in the literature.)
In particular, \citep{chen2021computing} and \citep{gao2022exact} proposed coloring and packing bounding rules, respectively, followed by \citep{chang2023efficient} and \citep{jin2024kd} that improve these bounds.
%\color{red} The coloring bound is based on the observation that any independent set of size $n$ has $\binom{n}{2}$ missing edges compared to a clique. Hence, if a graph is partitioned into several independent sets, one can easily derive a bound on the maximum number of independent sets that can form a $k$-defective clique.
It should be mentioned that \citep{wunsche2021mind} studied the crown decomposition rule to reduce search space, but only from a theoretical point of view. 
%Although its complexity is higher, it is a method of reducing the search space.

    Despite a considerable number of existing algorithms, we noticed that some important properties of this problem are overlooked. 
    %in the design of both theoretically and practically efficient branching algorithms.
    Firstly, existing known branching algorithms are built from scratch, rather than making use of the well-established algorithmic techniques like the maximum clique algorithm. In fact, there is a close relationship between a $k$-defective clique and a clique, that is, a $k$-defective clique is always a superset of a clique. It is possible to introduce the maximum clique algorithm for improving search efficiency. 
    Second, let us say that two vertices are conflict if they cannot be in the same maximum $k$-defective clique. 
    Given two conflict vertices $u$ and $v$, one can reason with this information to derive an tight upper bound of the solution. %containing both $u$ and $v$.
    %If we can incorporate this information into the branching algorithm, 
    As a result, there is a chance of further reducing the search space of the algorithm.
    However, to our knowledge, the pairwise \textit{conflict relations} of the vertices are not exploited. 

    Based on the above observations, we make the following contributions in the paper.
\begin{enumerate}
    \item \textbf{We design a new branching algorithm with a better worst-case runtime bound.}  
    Our algorithm is based on the structural property that a $k$-defective clique can be partitioned into a so-called \textit{$k$-defective set} and a clique.
    In a $k$-defective set, every vertex is not adjacent to at least one other vertex.
    The algorithm enumerates these $k$-defective sets and, for each such set $D$, calls the well-studied maximum clique algorithm to find the clique $C$ such that $D\cup C$ is a $k$-defective clique.
    When $k$ is a constant, the complexity of the algorithm is $O^*({\gamma_c}^n)$, where $\gamma_c$ is a maximum clique constant independent from $k$, e.g. ${\gamma_c}=1.20$ by \citep{xiao2017exact}. 
    We notice that this simple approach leads to a better complexity than all known algorithms when $k$ is a constant.
    We further introduce a graph decomposition technique to reduce the runtime bound to $O^*({\gamma_c}^{d(G)})$, where $d(G)$ is a practically small parameter called graph degeneracy.
  
    \item \textbf{We explore tighter upper bounds.} 
    We introduce a novel method to effectively compute an upper bound for the problem.
    Specifically, we propose an upper bounding rule which incorporates the property of pairwise conflict vertices. (Two vertices are conflict if they cannot be included in a sufficiently large $k$-defective clique.)
    We also design a dynamic programming-based algorithm to approximate the upper bound following this rule.
    %Based on this, we propose an optimization problem that incorporates the existing bounds and the conflict relations. 
    %We further explore the conflict relations between two vertices in the optimal $k$-defective clique and tighten the packing-and-coloring bound.
    By both theoretical and practical analysis, we show that the bound obtained by our algorithm beats other state-of-the-art bounding techniques. 
    To our knowledge, it is also the first time that a theoretical dominance relationship among different bounds has been established.
\end{enumerate}
 
Extensive experiments demonstrate that our algorithm outperforms the state-of-the-art algorithms across three sets of large real-world graphs. 
We also analyze different bounds empirically, showing that the new bound can prune much more branches than the existing ones. 

\section{Preliminary}

In an undirected graph $G=(V,E)$ where $V$ is the vertex set and $E$ is the edge set, $N_G(u)$ denotes the set of vertices adjacent to $u\in V$ in $G$ (when the context is clear, we use $N(u)$ instead),  $G[S]$ denotes the subgraph induced by $S\subseteq V$, $CN(S)$ denotes the set of vertices in $G$ that are adjacent to every vertex in $S$, that is, $\bigcap_{u\in S}N(u)$.
The complement graph of $G$ is denoted as $\overline{G}$. 
%$\overline{G}$ has the same vertex set as $G$, but 
Any two distinct vertices in $\overline{G}$ are adjacent if and only if they are not adjacent in $G$.
Given a positive integer $n$, we use $[n]$ to denote the index set $\{1,2,...,n\}$.
%When the context is clear, we use $n$ to represent the number of vertices in the input graph $G$.

A \textit{clique} in $G=(V,E)$ is a set $S \subseteq V$ where all vertices in $S$ are pairwise adjacent in $G$. 
An \textit{independent set} is a set $S \subseteq V$ where no two vertices in $S$ are adjacent in $G$.
Given a non-negative integer $k$, a set $S\subseteq V$ is a $k$-defective clique if the subgraph $G[S]$ contains at least $\binom{|S|}{2}-k$ edges. Equivalently, $S$ is a $k$-defective clique in graph $G$ if the complement graph $\overline{G[S]}$ has at most $k$ edges.
If $S$ is a $k$-defective clique in $G$, then any subset of $S$ is also a $k$-defective clique in $G$ \citep{gschwind2018maximum}.
For any subset $S \subseteq V$, we define $r(S) = k - |E(\overline{G[S]})|$ as the gap between $k$ and the number of edges in $\overline{G[S]}$. If $r(S) \geq 0$, then $S$ is a $k$-defective clique.

%We note the following key properties for the $k$-defective clique.
\begin{lemma}[Two-hop Property \citep{chen2021computing}]
\label{lemma_twohop_prop}
    If $S$ is a $k$-defective clique in $G$ with $|S|\ge k+2$, then $G[S]$ is a connected graph, and the length of the shortest path between any two vertices in $S$ is within 2.
    %is either adjacent to each other or adjacent to one other vertex in $S$ 
\end{lemma}
%The first property is referred to as the \textit{hereditary property} \citep{pattillo2013clique}, while the second property is known as the \textit{two-hop property} \citep{chen2021computing}. \
A $k$-defective clique with fewer than $k+2$ vertices may be disconnected. For example, an isolated vertex and a $k$-clique form a $k$-defective clique of size $k+1$. 
%In the paper, we concentrate on connected $k$-defective cliques.
Because we are only interested in the connected and large $k$-defective cliques, we investigate the \textit{non-trivial maximum $k$-defective clique problem} in the remaingings.
\begin{definition}[Non-trivial Maximum $k$-Defective Clique Problem]
    Given a graph $G$ and a non-negative integer $k$, find the largest $k$-defective clique $Q$ from $G$ such that its size is at least $k+2$. If there is no such $k$-defective clique, return 'no'.
\end{definition}
In practice, almost all large real-world graphs have non-trivial maximum $k$-defective cliques, e.g. 113 out of 114 Facebook graphs \footnote{\url{https://networkrepository.com/socfb.php}} in the renowned network repository \citep{nr-aaai15} have non-trivial $k$-defective clique when $k\le 5$.
%\wzr{Reference needed for network repository}

\section{An Exact Decompose-and-Branch Algorithm}

This section introduces a branching algorithm with a better exponential-time bound than existing ones.
We first introduce the notion of the $k$-defective set.
% it is clear that there are at most $k$ edges in $\overline{G[S]}$).
\begin{definition}[$k$-Defective Set]
    A $D\subseteq V$ is a called a $k$-defective set of $G$ if $\forall u\in D$, $|N(v)\cap D|<|D|-1$, and there are at least $\binom{|D|}{2}-k$ edges in $G[D]$.
    %then we call $D$ a $k$-defective set. 
\end{definition}
In general, a $k$-defective set $D$ is also a $k$-defective clique, except that every vertex in $D$ is nonadjacent to at least one other vertex in $D$. It is easy to see that the size of any $k$-defective set is no more than $2k$.

For any $k$-defective clique $Q$, we denote $D(Q)$ as the \textit{maximal $k$-defective set} in $Q$, that is, we cannot find another vertex $v\in Q$ such that $D(Q)\cup \{v\}$ is still a $k$-defective set.
Our algorithm is based on the observation that $Q$ can be partitioned into the maximal $k$-defective set $D(Q)\subseteq Q$ and a clique $C(Q)\subseteq Q$. 
%We denote $D(Q)$ as the maximal , which indicates that we cannot find another vertex $v\in Q$ such that $D(Q)\cup \{v\}$ is still a $k$-defective set.
We also observe that $C(Q)\subseteq CN(D(Q))$, that is, $\forall u\in D(Q)$, $C(Q)\subseteq N(u)$.
If $D(Q)$ is empty, then $Q$ is a clique in $G$. 
Examples of this observation are given in Fig. \ref{fig: kdc-structure}. 
%\wzr{Replace this maximal definition with bipartite}

\begin{figure}[htbp]
% \vspace{-2mm}
\centering
\includegraphics[width=0.8\columnwidth]{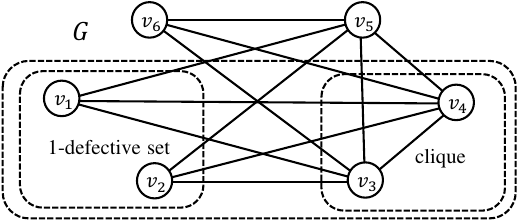}
\caption{
In graph $G$ of 6 vertices, for $S=\{v_1,v_2,v_3,v_4\}$, we have $G[S]$ is a $1$-defective clique with $D(S)=\{v_1,v_2\}$ and $C(S)=\{v_3,v_4\}$.
}
\label{fig: kdc-structure}
\end{figure}

%\wzr{modified, consider add shortcut notation for \( \bigcap_{ u \in D} N(u)\)}

\subsection{The Branching Algorithm}

By the above observation, our approach is outlined below. 
We enumerate all \(k\)-defective sets \(D\) from \(V\). For any non-empty \(k\)-defective set \(D\), we compute the maximum clique \(C^*\) within the subgraph induced by $CN(D)$. The maximum \(k\)-defective clique is the largest one among all candidates \(D \cup C^*\). 

For clarity in the description of the algorithm, we define the \textit{general $k$-defective clique problem}.

\begin{definition}[General $k$-Defective Clique Problem]\label{def:kdsEnum}
Given an instance $I=(G=(V,E),P,R)$ where $G$ is the graph, $P$ and $R$ are two disjoint subsets of $V$, find the maximum $k$-defective clique $Q^*$ such that $P \subseteq D(Q^*) \subseteq P\cup R$ and $Q^*\subseteq V$.
\end{definition}

We use $\omega_k(I)$ to denote the size of the maximum $k$-defective clique of instance $I$.
When the context is clear, we use $n$ to represent the number of vertices in $G$.

%In the implementation, we combine the two phases in the branching scheme. 
In Alg. \ref{alg: branch_framework}, we implement the above idea to solve the general $k$-defective clique problem with instance $I$.
For any input $I$, we find a $k$-defective set $D$ such that $P\subseteq D\subseteq P\cup R$. For a non-empty set $D$, we find the maximum clique $C$ from the set $CN(D)$. 
If $D$ is empty, we find the maximum clique from $C$ from $V$.
Then, the maximum $k$-defective clique is the largest set among all $D\cup C$. 
The following example explains this tree search algorithm.

%{\color{red}[We don't exiplitly define branchign rule and reduction rule here]
%In the following, given an instance $I=(G,P,R)$, we describe the \textit{branching rule} of our branching algorithm separately.
%A branching rule determines a vertex $u \in R$ such that $I$ can be branched into two instances that $u$ is in or not in the optimal solution.

\begin{figure}[htbp]
% \vspace{-2mm}
\centering
\includegraphics[width=\columnwidth]{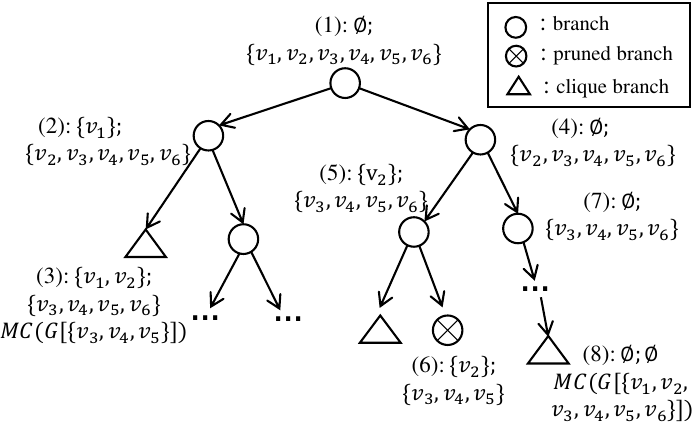}
\caption{The search tree of Alg. \ref{alg: branch_framework}. 
Each node in the tree represents a branch.
 The two sets labeled after the node represents $P$ and $R$, respectively, when the branch is generated.}
\label{fig:branch-emp}
\vspace{-2mm}
\end{figure}

\begin{algorithm}[ht!]
    \DontPrintSemicolon
        \caption{The branching algorithm for the maximum $k$-defective clique problem based on $k$-defective set enumeration.}
        \label{alg: branch_framework}
        \KwIn{An instance $I=(G,P,R)$ and and a known solution $Q^*$}
        \KwOut{A maximum $k$-defective clique $Q^*$ of $I$}   
       \textbf{Procedure}  branching($I=(G=(V,E),P,R)$)\\            
       \Begin{   
            \If{$\exists u\in P$ that $u$ is adjacent to all vertices in $P\cup R$ or $r(P)<0$}{
                \tcc{$P$ is not included in any $k$-defective set. }
                \Return
            }
            \If{$r(P) = 0$ or $R= \emptyset$}{
                \tcc{$P$ is a $k$-defective set and no more branching} 
                Let $C$ be $CN(P)$ if $P$ is non-empty, otherwise $V$\\
                $Q\gets \text{MC}(G[C])\cup P$\\
                \If {$|Q| > |Q^*|$}{
                    $Q^* \gets Q$
                }
                \Return
            }
            \If{$\exists v\in R$ that $u$ is adjacent to all vertices in $P\cup R$}{
               % \tcc{Reduce $v$}
                branching($I'=(G, P, R\setminus \{v\})$)\\
            }
            \Else{
                %\tcc{Need to branch}
                Find a branching vertex $v$ in $R$ such that $v=\arg\max_{u\in R}(|P|-|P\cap N(u)|)$\\
                branching($I'_1=(G,P\cup \{v\}, R\setminus \{v\})$)\\
                branching($I'_2=(G,P, R\setminus \{v\})$)\\
            }
        }
\end{algorithm}

\paragraph{Example}
Assume that the graph $G$ in Fig. \ref{fig: kdc-structure} is the input graph and $k=1$.
In Fig. \ref{fig:branch-emp}, we present the branching tree of Alg. \ref{alg: branch_framework}. Let us explain some representative branches.
\begin{itemize}
    \item In branch (1), $P$ is initialized as $\emptyset$ and $R$ is $V$. By the branching rule at line 14 in Alg. \ref{alg: branch_framework}, $v_1$ is selected as the branching vertex, and branches (2) and (4) are produced.
    \item In branch (2), by line 14 in Alg. \ref{alg: branch_framework}, $v_2$ is the branching vertex, two subbranches are produced.    
    \item In branch (3), because $r(\{v_1,v_2\})=0$, by lines 6-9 in Alg. \ref{alg: branch_framework},  then the $MC(\cdot)$ is called to find the maximum clique in $CN(P)=\{v_3,v_4,v_5\}$. 
    \item In branch (6), $P=\{v_2\}$ and $v_2$ is adjacent to $v_3,v_4,v_5$. Because there is no $k$-defective set $D$ such that $\{v_2\}\subseteq D\subseteq \{v_2,v_3,v_4,v_5\}$, by lines 3-4 in Alg. \ref{alg: branch_framework}, branch (6) is pruned.
    % \item For the pruned branch: In branch (6), as $v_2$ is adjacent to $v_3,v_4,v_5$, then there not existing a $k$-defective set $D$ such that $\{v_2\}\subseteq D\subseteq \{v_1,v_2,v_3,v_4\}$, so branch (6) should be pruned.
    \item In branch (7), $\exists v_3\in R$ that $v_3$ is all adjacent to $P\cup R$, by lines 11-12 in Alg. \ref{alg: branch_framework}, $v_3$ is reduced in the subbranches. 
    \item In branch (8), $R=\emptyset$, by line 6 in Alg. \ref{alg: branch_framework}, $MC(\cdot)$ is called to find the maximum clique from $V$.
\end{itemize}\
%{\color{red}[reconsider branch or node? In the figure, it's branch.]}
    % {\color{red}[Make the sentence in a uniform style. Cover all the cases in Alg. 1. Maybe 5 nodes?]}

One can simply justify that the branching algorithm enumerates all $k$-defective sets that subsume $P$ from $P \cup R$ without repetition. Now we estimate the complexity of the branching algorithm.
%Specifically, the algorithm enumerates all the $k$-defective sets from $I$ and then calls the maximum clique algorithm.

%and the upper bound of common neighbors in $V\setminus P$ is $C_{max}$, we have following lemma.
\begin{lemma}\label{lemma:kdsEnum}
Given an instance $I=(G, P,R)$ where $P\neq \emptyset$, the branching algorithm finds the maximum $k$-defective clique of $I$ in time $O(|R|^{2k}\gamma_c^{m})$ where 
$m=|CN(P)|$ and $\gamma_c$ is the base factor in the time complexity for maximum clique algorithm. 
\end{lemma}

\begin{proof}[Proof Sketch]
% \textit{Proof Sketch. }
Because the size of a $k$-defective set is bounded by $2k$, there are at most $|R|^{2k}$ different $k$-defective sets.
%$2k$ vertices can be moved from $R$ to $P$, i.e. $|R|^{2k}$ combinations in total.
    %\textcolor{blue}{It is obvious that one missing edge is corresponding to two vertices. Then the size of a $k$-defective clique is not larget than $2k$. 
    %Therefore, to enumerate each subset whose size is not larger than $2k$ in $R$ is $O(|R|^{2k})$. 
   %In the leaf branch, as shown in Fig .\ref{fig:branch-emp}, 
For each $k$-defective set $D$ that $ D \supseteq P$, we find the maximum clique algorithm in $CN(D)\subseteq CN(P)$ with a time complexity of $O^*(\gamma_c^{|CN(P)|})$. So, the complexity of the branching algorithm is $O^*(|R|^{2k}\gamma_c^{|CN(P)|})$.    
%For each $k$-defective set $subseteq P$, we call the maximum clique algorithm for $ CN(P)$ whose time complexity is $O^*(\gamma_c^{|CN(P)|})$. So, the complexity of the branching algorithm is $O^*(|R|^{2k}\gamma_c^{|CN(P)|})$.  
% \wzr{I will deal with it later.}
%The rigorous proof is left in the supplement.
\end{proof}

To our knowledge, $\gamma_c$ can be as small as $1.20$ by \citep{xiao2017exact}.
In our experiments, we use the maximum clique algorithm in \citep{chang2019efficient}. As far as we know, this is one of the most practically efficient maximum clique algorithms in large real-world graphs, despite there being no analysis of its time complexity. 
%If we use a theoretically efficient maximum clique algorithm like the Bron-Kerbosch \citep{tomita2006worst}, then we $\gamma_c=1.443$, already better than the bases of exponential complexity of any other maximum $k$-defective cliques.

\subsection{The Whole Decompose-and-Branch Algorithm}
We further investigate the \textit{graph decomposition rule} in the algorithm. 
The rule reduces the original input instance $I=(G,\emptyset,V)$ to instances $I_1,...,I_{n}$  such that the solution of $I$ is the maximum size in these $\omega_k(I_i)$. 

Because the rule is defined with an order of set $V$, let us first denote the order as $v_1,v_2,...,v_n$. We also denote $N(v_i)\cap \{v_{i+1},...,v_n\}$ as $N^+(v_i)$, and $\left(\bigcup_{v\in N^+(v_i)}N(v)\right)\cap \{v_{i+1},...,v_n\}\setminus N^+(v_i)$ as $N^{2+}(v_i)$. ($N^{2+}(v_i)$ is a subset of vertices that rank after $v_i$, and each $u\in N^{2+}(v_i)$ is adjacent to at least one vertex in $N^+(v_i)$). 
We say that $N^{2+}(v_i)$ is two-hop adjacent to $v_i$ with respect to the given order.

\begin{lemma}
Give an instance $I=(G,\emptyset,V)$ and an order of $V$ which is denoted by $v_1,...,v_n$. 
Then $\omega_k(I)=\mathop{\max}_{i=[n]}\left(\omega_k(I'_i),\omega_k(I''_i)\right)$ where $I'_i=(G[\{v_i\} \cup N^+(v_i)\cup N^{2+}(v_i)],\{v_i\}, N^+(v_i)\cup N^{2+}(v_i))$ and $I''_i=(G[\{v_i\}\cup N^+(v_i)],\emptyset,N^+(v_i))$.   
\end{lemma}
To understand this lemma, one can think that in a non-trivial $k$-defective clique, the length of shortest path of any two vertices is bounded by 2. So we can find the largest $k$-defective clique from all diameter-two bounded subgraphs.

Clearly, a order with small $|N^+(v_i)|$ and $|N^{2+}(v_i)|$ is preferred. 
We introduce the \textit{degeneracy ordering}. 
Given a graph $G$ with $n$ vertices, the degeneracy ordering is an order $v_1, v_2,...,v_n$ such that the largest $|N^+(v_i)|$ for all $i\in [n]$ is minimized. 
A degeneracy ordering of $G=(V,E)$ can be computed in time $O(|V|+|E|)$ by repeatedly removing a vertex with the minimum degree in the remaining graph until the graph becomes empty \citep{matula1983smallest}. 
%The degeneracy ordering technique minimizes the $|V^1_i|$ bound.
Given a degeneracy ordering, $|N^+(v_i)|$ is bounded by the \textit{degeneracy of $G$} which is normally denoted as $d(G)$.
Note that the degeneracy ordering has also been investigated in many maximum clique and near-clique problems  \citep{chang2019efficient,zhou2021improving,wang2023gap}.

%\subsection{The Whole Algorithm}
Now, we integrate the decomposition rule with our branching algorithm and obtain the whole decompose-and-branch algorithm DnBk in Alg. \ref{alg: 2hop_branch_framework}. 
%we illustrate the framework of our branch algorithm.
In line 3, we use a linear-time heuristic algorithm like that in \citep{chang2023efficient} to find an initial $k$-defective clique (see the supplementary for a complete description of this heuristic).
This initial solution is used to prune branches using our upper bounding rule in the following section. 
In lines 7 and 9, we call the branching algorithm with instances $I'_i$ and $I^{''}_i$, respectively.

\begin{algorithm}[ht!]
    \DontPrintSemicolon
        \caption{The decomposition-and-branching algorithm for maximum $k$-defective clique}
        \label{alg: 2hop_branch_framework}
        \KwIn{A graph $G$ and an positive integer $k$, where $n=|V|$}
        \KwOut{A maximum $k$-defective clique $Q^*$ in $G$}
        \textbf{Procedure} DnBk($G, k$)\\
        \Begin{
            Use a linear-time heuristic to find an initial  solution $Q$ from $G$\\
            Order $V$ as $v_1,...,v_n$ by degeneracy ordering \\
            \For{$i \in [n]$}{
                %$S\gets \emptyset$\\
                Build $I'_i=(G[\{v_i\} \cup N^+(v_i)\cup N^{2+}(v_i)],\{v_i\}, N^+(v_i)\cup N^{2+}(v_i))$\\
                %Let $G_i=G[\{v_i\} \cup N^+(v_i)\cup N^{2+}_{G}(v_i)]$,                  $\ P_i\gets \{v_i\},\ R_i\gets N^+(v_i)\cup N^{2+}_{G}(v_i) $\\
                branching$(I'_i)$\\
                %$G_i\gets G_i[N^+(v_i)\cup \{v_i\}],\ P_i\gets \emptyset,\ R_i\gets N^+(v_i)$\\
                Build  $I''_i=(G[\{v_i\}\cup N^+(v_i)],\emptyset,N^+(v_i))$ \\
                branching$(I''_i)$\\
            } 
            \Return $Q^*$
        }       
\end{algorithm}

%\paragraph{Time Complexity Analysis}
\begin{theorem}
    Given a graph $G$ and a constant non-negative integer $k$, denoting the degeneracy of $G$ as $d(G)$, the DnBk algorithm finds the maximum $k$-defective clique in time $O^*({\gamma_c}^{d(G)})$, where $\gamma_c$ is the base of the exponential factor in the time complexity for maximum clique algorithm. 
\end{theorem}
\begin{proof}
    First, let us consider instance $I'_i$ where $P=\{v_i\}$ and $R=\{N^+(v_i)\cup N^{2+}_{G}(v_i)\}$. Clearly, we have $|R|\leq \min(d(G)\Delta(G)+d(G), |V|)$. On the other hand, because $v_i\in P$, so $CN(P)$ is bounded by $|N^+(v_i)|$ and $|N^+(v_i)|\leq d(G)$. By Lemma \ref{lemma:kdsEnum}, the time complexity of the branching algorithm with $I'_i$ is $O((d(G)\Delta+d(G))^{2k+O(1)}\gamma_c^{d(G)})$.
    Second, consider instance $I''_i$ where $G=G[\{v_i\}\cup N^+_G(v_i)]$, $P=\emptyset$ and $R=N^+(v_i)$. Clearly, we have $|R|\leq d(G)$ and $|CN(P)|\le |\{v_i\}\cup N^+_G(v_i)|\le d(G)+1$. Consequently, the time complexity of the branching algorithm with $I''_i$ is $O(d(G)^{2k+O(1)}\gamma_c^{d(G)+1})$.
    Because $k$ is a constant, the total time complexity of DnBk is $O^*(\gamma_c^{d(G)})$.
\end{proof}
%An detailed description of correctness of this algorithm relies on an involved
\if 0
\begin{table}[htbp]
  \centering
  \caption{\color{red}[This table is not necessary!]A comparison of the time complexity for maximum $k$-defective clique computation and maximal $k$-defective clique enumeration algorithms that have nontrivial time complexities}

    \resizebox{\columnwidth}{!}{
         \begin{tabular}{ccc}
    \toprule
    Algorithm & Time Complexity & Problem \\
    \midrule
    MADEC+ \citep{chen2021computing} & $O^*(\gamma_{2k}^n)$ & Maximum $k$-defective clique \\
    Pivot2+ \citep{dai2023maximal} & $O^*(\gamma_{k}^d)$ & Maximal $k$-defective clique \\
    kDC \citep{chang2023efficient}   & $O^*(\gamma_{k}^n)$ & Maximum $k$-defective clique \\
    kDC-2 \citep{chang2024maximum} & $O^*(\gamma_{k-1}^d)$ & Maximum $k$-defective clique \\
    DnBk(This paper)  & $O^*(\gamma_c^d)$ & Maximum $k$-defective clique \\
    \bottomrule
    \end{tabular}%
    }
  \label{tab:time-complexity-cmp}%
\end{table}%
\fi
 \paragraph{Remark}
 Under the assumption that $k\ge 1$ is constant, this time complexity improves the known results from two perspectives. 
 First, the exponential base $\gamma_c$, which can be as small as 1.2 so far, is better than the state-of-the-art \citep{chang2024maximum}.
 %In theory, $\gamma_c$ could be as small as $1.20$ \citep{xiao2017exact}. In practice, the simple Bron-Kerbosch clique algorithm has $\gamma_c$ equal to $1.45$.
Second, the parameter (exponent) $d(G)$ is also smaller than $n$ in large graphs. 
For example, in soc-Livejournal, $d(G)$ is $213$ in contrast to $n=4033137$. 

\section{Exploring Tighter Upper Bound}
We investigate the \textit{upper bounding rule} to prune the branching algorithm. 
Given an instance $I=(G=(V,E),P,R)$, the bounding rule introduces an upper bound on $\omega_k(I)$. 
Suppose that there is currently a best known lower bound, say $lb$. (Initially, $lb=|Q|$ where $Q$ is found using a linear-time heuristic.)
In the branching algorithm, if the upper bound of $I$ is not greater than $lb$, we can claim that there is no hope of finding a $k$-defective clique larger than $lb$ in $I$, and thus we can stop the search. 
This is the principle of bounding in branch-and-bound algorithms. 
%In this section, we explore the upper bounds.

\subsection{Revisiting the Existing Upper Bounds} 

We first review the \textit{packing bound} and \textit{coloring bound} rules that are independently proposed in \citep{gao2022exact} and \citep{chen2021computing}.

    Before that, in an instance $I=(G,P,R)$, we define the \textit{weight of $u\in V$} as $w(u)=|P|-|N(u)\cap P|$.
    The $w(u)$ can be seen as the number of vertices that are not adjacent to $u$ in $P$. 
    Clearly, we can move some vertices of $V\setminus P$ to $P$. 
    If the sum of weight does not exceed $r(P)$, $P$ is still a $k$-defective clique. 
    This is the essence of packing bound.
    
    \begin{bounding}[Packing Bound  \citep{gao2022exact}]
    \label{bounding_packing}
        Assume that the vertices in $V\setminus P$ are sorted in non-decreasing order of weight $w(v)$, say $v_1,....,v_{|V\setminus P|}$. Then, $\vert P\vert + i$ is an upper bound of $\omega_k(I)$ where $i$ is the largest number that $\sum_{j=1}^i w(v_j)\leq k$.  
    \end{bounding}
    
%\paragraph{Coloring Bound}     
     The coloring bound depends on a graph coloring partition of $V\setminus P$. That is to say, $V\setminus P$ should be partitioned into $\chi$ non-disjoint independent sets before computing the coloring bound. 
     \begin{bounding}[Coloring Bound \citep{chen2021computing}]
     \label{bounding_coloring}
        Assume $V\setminus P$ is partitioned into $\chi$ non-disjoint independent sets $\Pi_1,...,\Pi_\chi$. Then $|P|+\sum_{j=1}^{\chi}min({\lfloor \frac{1+\sqrt{8k+1}}{2} \rfloor},{\vert \Pi_j\vert })$ is an upper bound of $\omega_k(I)$.   
     \end{bounding}
     This bound holds because at most $\lfloor \frac{1+\sqrt{8k+1}}{2} \rfloor$ vertices in an independent set can be contained in the same $k$-defective clique.

Recently, a \textit{Sorting} bound was proposed in \citep{chang2023efficient}, and another partition-based bound called \textit{Club} was given in \citep{jin2024kd}. These bounding rules are more complicated than the packing and coloring bounds. 
%We will make a dominance comparison between these bounds in the last part of this section.
In the following, we will investigate the \textit{conflict relationship} of the vertices and design a new bounding framework that hybrids the packing, coloring bounds and the conflict relationship.
%of a given instance, and then, into consideration to get a new Packing and Coloring Bound with Conflict Vertices, abbreviated as \textit{PackColorConf}. 

%Due to the space limit, the steps in detail of these two bounds are left in our appendix.

\subsection{A New Upper Bound Based on Packing, Coloring and Conflict Vertices}

Let us first introduce  \textit{conflict vertices}.
\begin{definition}[Conflict Vertices]
Give an instance $I=(G,P, R)$ and a lower bound $lb$, two different vertices $\{u,v\}$ ($u,v\in V$) are conflict vertices if $u$ and $v$ cannot be in the same $k$-defective clique of size larger than $lb$.    
\end{definition}

There are several rules to identify the conflict vertices for a given instance. So far, we use the following rules.
\begin{enumerate}
    \item If $u\in R$, $v\in V\setminus (P\cup R)$ and $\{u,v\}\notin E$, then $\{u,v\}$ are   conflict vertices.
   \item If $u,v\in V\setminus (P\cup R)$ and $\{u,v\}\notin E$, then $\{u,v\}$ are conflict vertices.    
    \item If $u,v\in V\setminus P$ and $r(P\cup \{u,v\})<0$, then $\{u,v\}$ are conflict vertices.
    \item If $u,v\in V\setminus P$, $\{u,v\}\in E$ and $|N(u)\cap N(v)\cap V\setminus P|\le lb-(|P|+r(P)-w(u)-w(v)+2)$, then $\{u,v\}$ are conflict vertices.
    \item If $u,v\in V\setminus P$, $\{u,v\}\notin E$ and $|N(u)\cap N(v)\cap V\setminus P|\le lb-(|P|+r(P)-w(u)-w(v)+1)$, then $\{u,v\}$ are conflict vertices.
\end{enumerate}

%Proofs of the correctness of the above rules, as well as the following lemmas and algorithms, are left in the appended files.

For each pair of vertices $u,v\in V\setminus P$, we check whether $\{u,v\}$ are conflict vertices. In each branch, we keep updating the conflict information. 
Note that more rules can possibly be discovered for future extension.
We use these rules because they are easy to implement and identifying them does not require too much computation effort.

In the following, we denote $conflict(u,v)=1$ if $u$ and $v$ are conflic vertices and $conflict(u,v)=0$ otherwise. Now, we are ready to formulate our bounding rule with this conflict information, namely, the \textit{packing, coloring with conflict bounding rules}.

\begin{bounding}[The Packing, Coloring with Conflict Bounding Rules]
\label{bounding_pack_color_exclusive}
    Given an instance $I=(G,P,R)$ and a lower bound $lb$, assume that $V\setminus P$ is partitioned into $\chi$ independent sets $\Pi_1,...,\Pi_\chi$. 
    For any $u,v\in V\setminus P$, $conflict(u,v)=1$ if $\{u,v\}$ are conflict vertices and $conflict(u,v)=0$ otherwise.    

     The optimal objective value of the optimization problem described in the following is an upper bound of $\omega_k(I)$.

\begin{align}
    \max_{S_i\subseteq \Pi_i,\forall i\in [\chi]}  & \quad |P|+\sum_{i=1}^{\chi}{|S_i|} & \text{ \textbf{OPT}} \nonumber\\
    \text{s.t.} &  \sum_{i=1}^{\chi}{\left( \binom{|S_i|}{2}+ \sum_{u\in S_i}w(u)\right)}\le r(P) \label{constraint_pack_color}\\
    & conflict(u,v)=0, \forall u,v \in S_i\ \forall i\in [\chi]   \label{constraint_conflict} 
\end{align}
\end{bounding}

The optimization problem OPT asks for a subset $S_i$ from each $\Pi_i$ such that $S_i$ meets the (\ref{constraint_pack_color}), the \textit{packing-and-coloring constraint} and (\ref{constraint_conflict}), the \textit{conflict constraint}.
When $\chi=1$, the OPT problem is a special case of the NP-hard knapsack problem with conflict graph \citep{bettinelli2017branch}. 
Hence, it is unlikely to efficiently solve the optimization problem due to its NP-hard nature. 

However, the optimal solution to the relaxed OPT problem is still a feasible bound of $\omega_k(I)$. 
In the following, we first introduce an algorithm to solve the relaxed OPT problem which excludes the conflict constraint (\ref{constraint_conflict}). 
Then, we improve this algorithm so that some of the conflict constraint \ref{constraint_conflict} can be satisfied. 
Clearly, the objective value is an upper bound of the OPT problem, which is in turn, an upper bound of $\omega_k(I)$.

 %\subsection{An Dynamic Programming  New Upper Bound}

\subsubsection{Dynamic Programming without Conflict Constraint}\label{subsec:packcolor}

We propose the following dynamic programming-like algorithm to solve the OPT without considering the conflict constraint.

\paragraph{DPBound Algorithm}
\begin{enumerate}   
    \item For each independent set $\Pi_i$ where $i\in [\chi]$, order the vertices in $\Pi_i$ by the non-decreasing order of the weight $w(\cdot)$. Assume that $\Pi_i$ is ordered as $v_i^1,v_i^2,...,v_i^{|\Pi_i|}$. For each $r \in \{0,...,r(P)\}$, find the maximum number $j\in \{0,...,|\Pi_i|\}$ such that $\binom{j}{2}+ \sum_{k=1}^{j}w(v_i^k)\le r$. Assign $t(i,r)=j$..
    \item For each independent set $R_i$ ($i\in [\chi]$) and $r\in \{0,...,r(P)\}$, define $f(i,r)$ as the optimal value when $\chi =i$ and $r(P)=r$. Then $f(i,r)$ can be computed by the following recurrence relations. 
    \begin{align*}
     f(1, r)=& t(1,r), \textbf{\ \ \ \ } \forall r\in \{0,...,r(P)\}\\   
     f(i, r)=& \displaystyle \max_{r'\in \{0,...,r\}}f(i-1,r')+t(i-1,r-r'), \\
     &\forall i \in \{2,...,\chi\} 
    \end{align*} 
    \item Return $|P|+f(\chi, r(P))$ as the optimal objective value.
\end{enumerate}

For the DPBound algorithm, the first step can be easily implemented in time $O(|V\setminus P|)$ using a linear-time sorting algorithm.
In the second step,  we use a standard bottom-up dynamic program to compute the relation, which has running time $O(\chi \cdot r(P)^{2})$. 
The total running time is $O(\vert V\setminus P\vert +\chi \cdot r(P)^{2})$.

% \fi
\paragraph{Example} 
We show an example of the DPBound algorithm in Fig. \ref{fig:bound-example}. 
In the example, $V\setminus P$ is partitioned into two independent sets $\Pi_1=\{v_1^1. v_1^2. v_1^3, v_1^4\}$ and $\Pi_2=\{v_2^1, v_2^2, v_2^3, v_2^4\}$, and each set is sorted in non-decreasing order of the weight $w(\cdot)$. 
The two tables on the left of the bottom show the results of $t(i,r$) and $f(i,r)$ by running the DPBound algorithm. 
Finally, the bound is $|P|+f(2,5)=8$.

\begin{figure}[htbp]
\vspace{-2mm}
\includegraphics[width=\columnwidth]{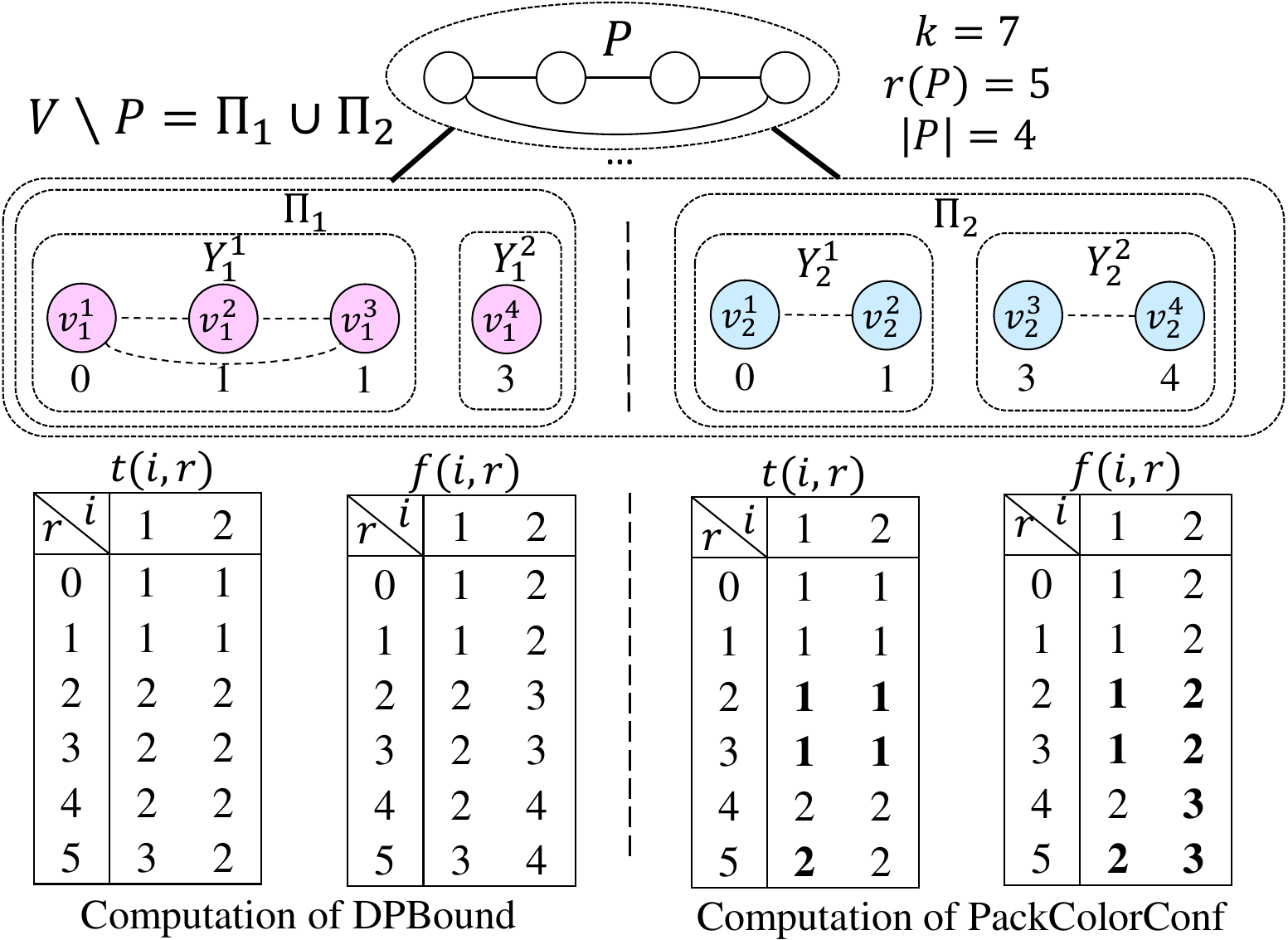}
\caption{
An example of running the DPbound with and without conflict vertices. The dashed lines represent conflict relations. 
The weight of vertices in $V\setminus P$, $w(\cdot)$, is displayed below the vertices.}
\label{fig:bound-example}
\vspace{-2mm}
\end{figure}

\subsubsection{The Algorithm for Packing, Coloring and Conflict Bounds}
\label{subsubsec-alg-packcolorconf}
Now, we propose an algorithm to solve OPT such that the packing-and-coloring constraint is satisfied and the conflict constraint is partially satisfied.

\begin{algorithm}[ht!]
    \DontPrintSemicolon
        \caption{Approximate the packing, coloring bound with conflict constraint.}
        \label{alg:algub2}
        \KwIn{An instance $I=(G,P,R)$, a known solution $Q'$, and a partition of $V\setminus P$ into independent sets $\mathbf{\Pi}=\{\Pi_1,...,\Pi_\chi\}$.}
        \KwOut{An upper bound of $\omega_k({I})$.}   
        PackColorConf($I=(G,P,R), lb, \mathbf{\Pi}$) \\
        \Begin{
            %\For{each pair $\{u,v\}$ in $R$}{
            %if $\{u,v\}$ is conflict vertices, $conflict(u,v)\gets 1$ , otherwise $conflict(u,v)\gets 0$
            %}
            Build up the conflict function $conflict(\cdot)$ for each pair of vertices $\{u,v\}$ in $V\setminus P$ \\
            \For{$\Pi_i$ from $\Pi_1$ to $\Pi_\chi$}{
                         
                $\sigma\gets 1,  \mathcal{Y}\gets \{\emptyset\}$\\       
                \While{$\Pi_{i}\neq \emptyset$}{
                    %$\sigma\gets \sigma+1$\\
                    Sort $\Pi_{i}$ by non-decreasing order of $w(\cdot)$\\
                    Build an empty set $Y_{\sigma}$\\
                    \For{each $v$ in $\Pi_{i}$ by the sorted order}{
                        \If{$Y_{\sigma}=\emptyset$, or $\forall u\in Y_{\sigma}\text{ } conflict(u,v)=1$ }{
                            $Y_{\sigma}\gets Y_{\sigma}\cup \{v\}$\\
                        }
                    }
                    $\mathcal{Y}\gets \mathcal{Y}\cup \{Y_\sigma\},\Pi_{i}\gets \Pi_{i} \setminus{Y_{\sigma}}$\\
                    $\sigma\gets \sigma   +1 $                    
                }                
                \For {$r$ from $0$ to $r(P)$}{    
                    %\If{$\min_{v\in Y_0}w(v)\leq r$}{
                    %    $flag\gets 1$
                    %}
                    Compute the maximum index $j\in [\sigma]$ such that $\binom{j}{2}+\sum_{l=1}^{j}\left(\min_{u\in Y_{l}}w(u)\right) \leq r$.\\
                    $t(i,r)\gets j$
                }
            }
        Use the linear recurrence in step 2 in the DPBound algorithm and compute $f(\chi,r(P))$\\
            \Return $|P|+f(\chi, r(P))$
        }
    % \vspace{-2mm}
\end{algorithm}

The general idea of PackColorConf is as follows. We partition each $\Pi_i$ into subsets $Y_1,...,Y_{\sigma}$ such that vertices in each set $Y_j$ ($j\in [\sigma]$) are mutually conlict (lines 6-13). 
%As a result, every pair of vertices in $Y_j$ ($j\in [\sigma]$) is a conflict pair.
Then, for each $\Pi_i$, we compute the values of $t(i,r)$ for each $r\in \{0,...,r(P)\}$ by line 15. 
If the partition of $\Pi_i$ is non-trivial (that is, a subset $Y_j$ has more than one vertex), we may obtain $t(i,r)$ values smaller than in the first step of PackColorConf, thus tightening the final bound.

\paragraph{Example}
    We show an example of Alg. \ref{alg:algub2} in Fig. \ref{fig:bound-example}. 
    Suppose that $\Pi_1$ is partitioned into two mutual conflict sets $Y_1^1$ and $Y_1^2$. 
    Taking the computation of $t(1,3)$ as an example. 
    Because $\arg\min_{u\in Y_1^1}=v_1^1,\arg\min_{u\in Y_1^2}=v_1^4$ and $w(v_1^1)+w(v_1^4)+1>3$ but $w(v_1^1)+0<3$, we can get $t(1,3)=1$ which is smaller than that $2$ obtained in the DPBound algorithm.
     After computing all the values of $t(i,r)$, we calculate the final bound $|P|+f(2,5)=7$. 
     This is smaller than $8$, the bound obtained by the DPBound algorithm.

\subsection{Comparison of Existing Bounding Rules}
\label{subsec_discuss_bound}
We compare the recent bounding rules for the maximum $k$-defective clique through the domination relations.
For a maximize problem, we say that \textit{Bounding Rule A dominates Bounding Rule B} if, for the same instance, the bound found by A is not larger than that found by B.
In Fig. \ref{fig:tight-compare}, we show the dominance relations among the packing bound \citep{gao2022exact}, coloring bound \citep{chen2021computing}, sorting bound \citep{chang2023efficient}, club bound \citep{jin2024kd} and the bound obtained by PackColorConf.
Note that PackColorConf only obtains the upper bound of the OPT problem (Section \ref{subsubsec-alg-packcolorconf}). 
%Detail proofs for each pairwise dominance relationship are left in the appended file.
We can see that PackColorConf is as good as Sorting bound, which in turn is as good as Coloring and Packing bounds. Due to the heuristic nature of the Club, we only know that it dominates Packing so far.

\begin{figure}[H]
\centering
% \vspace{-2mm}
\includegraphics[width=0.95\columnwidth]{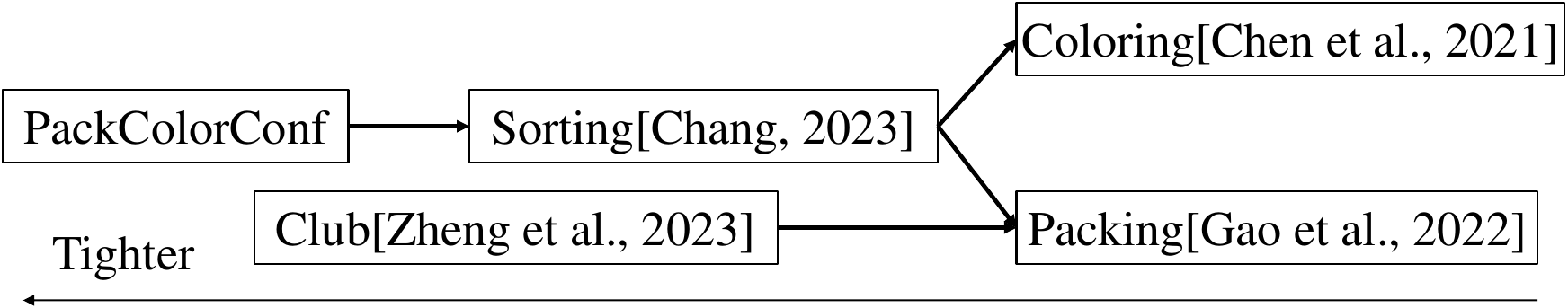}
\caption{The dominance relations among bounding rules. $A \longrightarrow B$ indicates that A dominates B.}
\label{fig:tight-compare}
\vspace{-2mm}
\end{figure}

\section{Experiments}
In this section, we evaluate our algorithms and bounds empirically.  
Our algorithm is written in C++11 and compiled by g++ version 9.3.0 with the -Ofast flag. 
All experiments are conducted on a machine with an Intel(R) Xeon(R) Gold 6130 CPU @ 2.1GHz and a Ubuntu 22.04 operating system. 
Hyper-threading and turbo are disabled for steady clock frequency.

We implement our DnBk branching algorithm that uses PackColorConf to upper bound the solution (and thus prune branches). 
\footnote{Codes and experiment data are available at \url{https://github.com/cy-Luo000/Maximum-k-Defective-Clique.git}.} 
We empirically compare DnBk with the state-of-the-art algorithms kDC \citep{chang2023efficient}, KDClub  \citep{jin2024kd}. 
Three real-world datasets are used as benchmark graphs, the same as in \citep{chang2023efficient,gao2022exact}.
Because the kDC-2 codes~\cite{chang2024maximum} are not publicly available, we only partially compare this algorithm using the data in the paper. This is reported in Section 6.2 in the supplementary file. 
%\textcolor{blue}{Because the codes of KDBB \citep{gao2022exact} and kDC-2 \citep{chang2024maximum} are not publicly available, we cannot compare our algorithm with them.}

\begin{figure*}[ht]
\vspace{-2mm}
\includegraphics[width=\linewidth,height=0.4\linewidth]{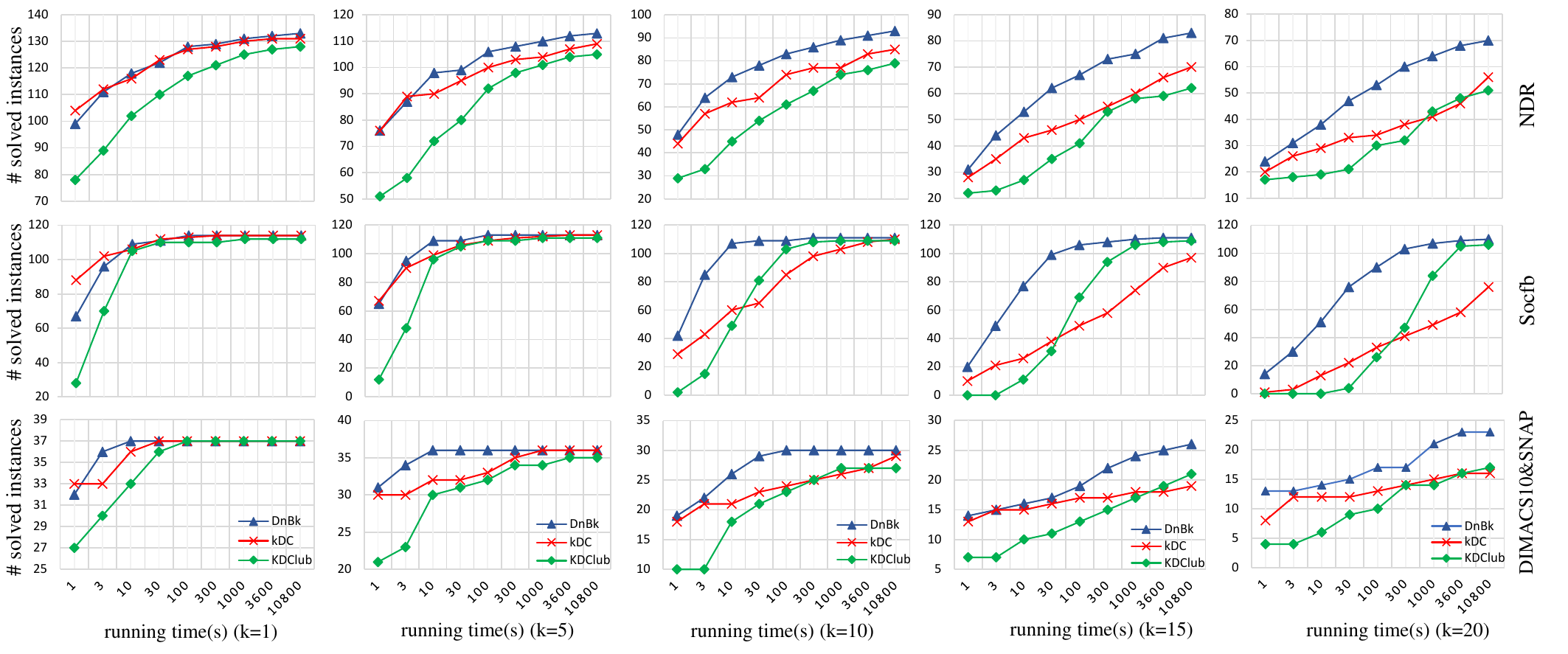}
\caption{
Number of solved instances for NDR, Socfb  and DIMACS10\&SNAP graphs, with $k=1,5,10,15,20$ and a cutoff time $10800$s.}
\label{fig:overall-performance}
% \vspace{-2mm}
\end{figure*}

%\paragraph{Datasets}
%We evaluated these algorithms on 
\begin{itemize}
    \item \textbf{The NDR} dataset contains 139 real-world graphs with up to $5.87 \times 10^7$ vertices from the Network Data Repository \footnote{\url{https://networkrepository.com/}}, including social networks, biological networks, collaboration networks and so on.
    \item \textbf{The DIMACS10\&SNAP}  dataset contains 37 graphs with up to $5.61 \times 10^6$ vertices, 27 of them are from the DIMACS10 competition\footnote{\url{https://networkrepository.com/dimacs10.php}}, and 10 are from SNAP  library \footnote{\url{http://snap.stanford.edu/data/}}.
    \item \textbf{The Socfb} set contains 114 graphs with up to $1.18\times 10^8$ vertices, extracted from Facebook social networks. \footnote{\url{https://networkrepository.com/socfb.php}}
\end{itemize}

We record the computation time for running an algorithm on a graph instance for a specific $k$ from $\{1,5,10,15,20\}$, with a cutoff time $10800$ seconds (3 hours) for each test. 
This configuration is the same as that in \citep{chang2023efficient}.
The recorded time excludes the I/O time of loading the graph instance from the disk to the main memory. 
%Similarly to \citep{chang2023efficient}, we choose $k$ from $\{1,5,10,15,20\}$ and set a time limit of $3$ hours for each test. 
Because the problem we solve is to find the maximum $k$-defective clique of size larger than or equal to $k+2$, we omit instances whose maximum $k$-defective clique size is smaller than $k+2$ for each $k$.

\subsection{Overall Performance}
%In this subsection, we evaluate our algorithm DnBk against the existing two most efficient algorithm kDC and KDBB. 
In Fig. \ref{fig:overall-performance}, we show the number of solved instances within different time frames. 
The codes for KDBB \citep{gao2022exact} are not available because the authors cannot provide them. 
So, it is not included in the figure.
In fact, both kDC and KDClub outperform KDBB in their reports. Our DnBk also solves all instances that were solved in \citep{gao2022exact} but consumes much less time.
%{\color{orange}KDBB can be simplified} 

From Fig. \ref{fig:overall-performance}, DnBk solves more instances or at least as many instances as kDC and KDClub in 3 hours for all sets. 
%For small $k$ values like $1$ and $5$, DnBk solves more instances than kDC and KDClub in many . 
For small $k$ values such as $1$ or $5$, these algorithms can solve almost all instances of DIMACS10\&SNAP and Socfb to optimal in 3 hours, but DnBk still solves more in a shorter time frame.
For $k=10,15$ and $20$, DnBk clearly outperforms kDC and KDClub. 
\if 0
\textcolor{blue}{From table \ref{tab:bound-cmp} we can also see that the degeneracy of a large realworld graph is much smaller than its size.}
\fi
%In general, with increasing $k$, there is a clear gap between DnBK and the other two in terms of the number of solved instances. 
Specifically, we find that DnBk closes 3,2,5,8 and 6 NDR instances that were not solved by other algorithms for $k=1,5,10,15$ and $20$, respectively. 

% \textcolor{blue}{We can also see that in Tab. \ref{tab:bound-cmp}, for 90\% instances the maximum $k$-defective clique is at least 2 times as large as $k$. Indeed, most large real-world graphs have non-trivial maximum $k$-defective cliques, e.g. 99.6\% Socfb instances for $k=1,5$ and 98.2\% instances for all tested $k$ values.}

% \tabcolsep=0.06cm
\tabcolsep=0.03cm
\begin{table}[H]
  \centering
  % \vspace{-2mm}
  \caption{The time (in seconds) and search tree nodes (branch) number (in $10^3$ scale) of differen algorithms on example graphs.
 We show cases where $k=1,10,20$.  \emph{opt} indicates the maximum solution size. \textit{OOT} means a failure due to out-of-time. The instances that cannot be solved by any algorithms within 10800s are omitted.} 
 \vspace{\topsep}
  %\emph{nodes} refers to the number of recursive calls made by branch-and-bound algorithm(the unit is $10^3$).
  % time is accurate to seconds.
    \resizebox{\columnwidth}{!}{
   \begin{tabular}{ccccccccccccc}
    \toprule
    % \toprule
    \multirow{2}[2]{*}{\textbf{Graph}} & \multirow{2}[2]{*}{\textbf{k}} & \multirow{2}[2]{*}{\textbf{opt}} & \multicolumn{5}{c}{\textbf{\#Nodes(10\textsuperscript{3})}} & \multicolumn{5}{c}{\textbf{Time(s)}} \\
\cmidrule(lr){4-8}  \cmidrule(lr){9-13}          &       &       & DnBk & DnBk\textsubscript{S} & DnBk\textsubscript{Club} & DnBk\textsubscript{P} & DnBk\textsubscript{C} & DnBk & DnBk\textsubscript{S} & DnBk\textsubscript{Club} & DnBk\textsubscript{P} & DnBk\textsubscript{C} \\
    \midrule
  \multirow{3}{*}{\tabincell{c}{ socfb-Auburn71 }}
  & 1     & 58    & \textbf{0.1} & 0.1   & 0.1   & 3.1   & 0.2   & 3.3   & \textbf{3.2} & 3.9   & 4.6   & 6.6 \\
  & 10    & 63    & \textbf{14.2} & 94.4  & 92.1  & 4902.4 & 1084.3 & \textbf{6.8} & 21.2  & 21.8  & 2389.0 & 712.2 \\
  & 20    & 67    & \textbf{1274.8} & 7720.3 & 10349.2 & OOT   & OOT   & \textbf{233.5} & 1329.9 & 1801.3 & OOT   & OOT \\
\midrule
\multirow{3}{*}{\tabincell{c}{ socfb-Tennessee }}
  & 1     & 59    & \textbf{0.0} & 0.0   & 0.0   & 0.6   & 0.0   & 2.1   & \textbf{2.0} & 2.3   & 3.3   & 3.9 \\
  & 10    & 66    & \textbf{0.6} & 16.4  & 16.5  & 1339.0 & 29.8  & \textbf{2.1} & 4.1   & 4.3   & 418.4 & 17.0 \\
  & 20    & 70    & \textbf{125.5} & 2045.0 & 2590.3 & 9870.6 & 4720.0 & \textbf{26.0} & 406.7 & 395.0 & 3395.2 & 789.5 \\
\midrule
\multirow{3}{*}{\tabincell{c}{ socfb-Texas84 }}
  & 1     & 52    & 0.2   & 0.2   & \textbf{0.1} & 6.4   & 0.3   & \textbf{4.1} & 4.5   & 5.2   & 5.9   & 8.0 \\
  & 10    & 58    & \textbf{8.4} & 40.0  & 38.0  & 4730.7 & 2230.3 & \textbf{6.2} & 11.9  & 14.5  & 2379.0 & 1641.0 \\
  & 20    & 61    & \textbf{3193.5} & 17647.5 & 22635.2 & OOT   & OOT   & \textbf{543.2} & 3182.7 & 3453.5 & OOT   & OOT \\
\midrule
\multirow{3}{*}{\tabincell{c}{ tech-WHOIS }}
  & 1     & 59    & 0.2   & 0.2   & \textbf{0.1} & 5.4   & 0.3   & 0.5   & \textbf{0.5} & 0.6   & 1.8   & 1.3 \\
  & 10    & 64    & \textbf{95.7} & 849.5 & 836.0 & OOT   & 3367.5 & \textbf{22.5} & 173.9 & 147.3 & OOT   & 3439.0 \\
  & 20    & 69    & \textbf{5667.7} & OOT   & OOT   & OOT   & OOT   & \textbf{1265.1} & OOT   & OOT   & OOT   & OOT \\
\midrule
\multirow{3}{*}{\tabincell{c}{ soc-Epinions1 }}
  & 1     & 24    & 0.4   & 0.4   & \textbf{0.3} & 5.5   & 1.0   & 5.2   & 5.2   & 7.1   & \textbf{4.6} & 10.3 \\
  & 10    & 29    & \textbf{131.7} & 507.2 & 511.0 & 13029.3 & OOT   & \textbf{22.2} & 51.6  & 54.3  & 3199.5 & OOT \\
  & 20    & 32    & \textbf{26184.3} & 41184.5 & 97583.9 & OOT   & OOT   & \textbf{2072.6} & 3433.7 & 5738.3 & OOT   & OOT \\
\midrule
\multirow{3}{*}{\tabincell{c}{ email-EuAll }}
  & 1     & 17    & \textbf{0.0} & 0.0   & 0.0   & 0.1   & 0.0   & \textbf{0.2}   & 0.2 & 0.3   & 0.2   & 0.4 \\
  & 10    & 21    & \textbf{39.0} & 61.8  & 80.4  & 1229.5 & 3375.9 & \textbf{4.1} & 5.0   & 6.2   & 117.1 & 371.9 \\
  & 20    & 24    & \textbf{1912.4} & 2034.8 & 9827.5 & OOT   & OOT   & \textbf{360.4} & 446.5 & 776.9 & OOT   & OOT \\
\midrule
\multirow{3}{*}{\tabincell{c}{ as-22july06 }}
  & 1     & 18    & \textbf{0.0} & 0.0   & 0.0   & 0.0   & 0.0   & 0.0   & \textbf{0.0} & \textbf{0.0} & 0.0   & 0.0 \\
  & 10    & 22    & \textbf{2.4} & 2.6   & 5.0   & 18.7  & 24.9  & \textbf{0.3} & 0.3   & 0.6   & 2.8   & 3.5 \\
  & 20    & 24    & \textbf{474.7} & 483.5 & 1591.8 & 14409.8 & OOT   & \textbf{47.9} & 61.9  & 115.5 & 1630.6 & OOT \\
\midrule
\multirow{2}{*}{\tabincell{c}{ soc-orkut }}
  & 1     & 48    & 0.5   & 0.5   & \textbf{0.4} & 261.5 & 1.9   & \textbf{718.7} & 768.7 & 782.2 & 1250.9 & 1026.8 \\
  & 10    & 53    & \textbf{157.3} & 422.7 & 313.5 & OOT   & OOT   & \textbf{936.8} & 1360.3 & 1324.8 & OOT   & OOT \\
\midrule
\multirow{2}{*}{\tabincell{c}{ rt-retweet-crawl }}
  & 1     & 14    & \textbf{0.0} & 0.0   & 0.0   & 0.0   & 0.0   & 0.7   & \textbf{0.6} & 0.6   & 1.5   & 1.1 \\
  & 10    & 17    & \textbf{0.0} & 0.0   & 0.8   & 4.6   & 5.9   & 0.9   & \textbf{0.7} & 0.9   & 1.8   & 2.1 \\
  % & 20    & 21    & \textbf{0.0} & 0.0   & 6.0   & 366.8 & 341.1 & 572.3 & \textbf{527.5} & 700.9 & 7019.4 & 3141.5 \\
\midrule
\multirow{3}{*}{\tabincell{c}{ tech-as-skitter }}
  & 1     & 68    & \textbf{0.0} & 0.0   & 0.0   & 0.0   & 0.0   & 1.9   & 1.8   & 1.9   & \textbf{1.7} & 3.3 \\
  & 10    & 72    & \textbf{1.1} & 1.2   & 1.8   & 13.6  & 9.3   & \textbf{4.0} & 4.4   & 4.7   & 26.2  & 15.6 \\
  & 20    & 75    & \textbf{24.0} & 24.8  & 54.5  & OOT   & OOT   & \textbf{20.0} & 20.9  & 32.9  & OOT   & OOT \\
\midrule
\multirow{2}{*}{\tabincell{c}{ soc-youtube }}
  & 1     & 17    & \textbf{0.0} & 0.0   & 0.0   & 0.1   & 0.0   & \textbf{2.6} & 2.7   & 3.2   & 2.9   & 4.9 \\
  & 10    & 22    & \textbf{48.5} & 55.2  & 66.5  & 1712.7 & 3300.9 & \textbf{33.4} & 36.1  & 40.7  & 722.7 & 1742.0 \\
\midrule
\multirow{3}{*}{\tabincell{c}{ socfb-UGA50 }}
  & 1     & 53    & 0.2   & 0.2   & \textbf{0.2} & 12.1  & 0.5   & 4.1   & \textbf{4.0} & 5.0   & 7.7   & 7.5 \\
  & 10    & 58    & \textbf{27.5} & 90.5  & 82.2  & OOT   & 5319.5 & \textbf{10.4} & 23.6  & 31.4  & OOT   & 4950.6 \\
  & 20    & 61    & \textbf{21325.2} & OOT   & OOT   & OOT   & OOT   & \textbf{3601.1} & OOT   & OOT   & OOT   & OOT \\

    % \bottomrule
    \bottomrule
    \end{tabular}%
    }
  \label{tab:bound-cmp}%
  \vspace{-2mm}
\end{table}%

\subsection{The Effect of the Upper Bound}
 We empirically evaluate the upper bounds by comparing DnBk (which uses packing-and-coloring bound with conflict pairs) with the following algorithms.
\begin{itemize}
    \item DnBk\textsubscript{P}, the DnBk that uses the packing bound, see Bounding Rule \ref{bounding_packing}. 
    \item DnBk\textsubscript{C}, the DnBk that uses the coloring bound, see Bounding Rule \ref{bounding_coloring}.
    \item DnBk\textsubscript{S}, the DnBk that uses the \textit{Sorting} Bound in \citep{chang2023efficient}.
    \item DnBk\textsubscript{Club}, the DnBk that uses the \textit{Club} bound in \citep{jin2024kd}.
\end{itemize}

%Note that the packing-and-coloring bound rule is equal to the sorting bound rule in\citep{chang2023efficient}. 
Due to the space limit, we show the computation time and the nodes of each algorithm for $12$ benchmark graphs in Tab. \ref{tab:bound-cmp}. 
Here, the number of nodes refers to the number of branches in the search tree produced by the algorithm.
We can see that for all instances, DnBk has fewer nodes than DnBk\textsubscript{S}, which in turn has fewer nodes than DnBk\textsubscript{P} and DnBk\textsubscript{C}. 
This matches the dominance relations that we conclude in Fig. \ref{fig:tight-compare}.
%Furthermore, there is no clear dominance relation in terms of tree size between DnBk\textsubscript{P} and DnBk\textsubscript{C} empirically. 
For instances like rt-retweet-cra1w and tech-as-skitter with $k$=1,  all algorithms can solve them within 10 seconds, the other variant could be faster than DnBk. 
This shows that the time overhead of computing the bound only affects easy instances.
The number of nodes is fewer than that of DnBk\textsubscript{Club} in most cases, implying that our new bound is empirically better than Club. Also, DnBk is faster than this algorithm.
%Lastly, we note that our algorithm DnBk which uses the bound utilizing the conflict pairs can be faster than all other variant algorithms without considering the conflict relationship between vertex pairs.

\section{Conclusion}
In the paper, we discussed both theoretical and practical aspects of solving the maximum $k$-defective clique problem. 
We studied a branching algorithm that has a better exponential time complexity than the existing ones when $k$ is constant.
We also investigated a new bound that hybrids existing two bounds with the idea of conflict vertices.
The tightness of the bounds are analyzed based on the dominance relationships.
%Firstly, we study a branching algorithm that exploits the $k$-defective sets and the two-hop property of the problem.
%As a result, the algorithm has a better exponential time complexity than the existing ones when $k$ is constant.
%Secondly, we investigate methods to compute the upper bound. We revisit existing bounds such as Packing and Coloring, and propose a new bound to hybrid them with the idea of conflict vertices.
%The tightness of the bounds are compared based on the dominance relationships.
%Lastly, experiments show the superior performance of the final algorithm as well as the effectiveness of the bound.
We note that the conflict relationship of vertices, which was first formalized in this paper, exists in many graph search problems, such as the maximum clique problem, and thus can be potentially extended in the future.

\section*{Acknowledgements}
This work was supported by Natural Science Foundation of Sichuan Province of China under grants 2023NSFSC1415 and 2023NSFSC0059, and National Natural Science Foundation of China under grants 61972070 and 62372095. The second author is also funded by China Postdoctoral Science Foundation under grant 2022M722815.

%\section*{Acknowledgments}

% \bibliographystyle{named}
% \bibliography{ijcai24}

% \newpage 

%%%%%%%%%%%%%%%%%%%%%%%%%%%%%%%%%%%%%%%%%%%%%%%%%%%%%%%%%%%%%%%%%%%%%%%%

%%% Use this command to include your bibliography file.

\bibliography{M749}

\end{document}

\section*{Appendix}
\section{Structure of the Appendix}

\begin{itemize}
    \item In Section \ref{section-heuristic-algorithm}, we describe the heuristic algorithm used in BnBk.
    \item In Section \ref{section-lemma3-decompositioncorrect}, we prove the correctness of Lemma 2 and Lemma 3. Specifically, we prove that the decomposition algorithm does not miss the optimal solution.
    \item In Section \ref{section-bound-algorithm}, we add the missing proofs for the upper bounding rules and algorithms in  the paper. Specifically, in Section 3.1, we show the correctness of all conflict rules in the paper. In Section 3.2, we show the correctness of the pakcing, coloring with conflict bounding rules. In Section 3.3, we show that DPBound solve the OPT problem without conflict constraint.
    \item In Section \ref{section-proof-dominance}, we prove the dominance relation among the proposed bounds and existing bounds.
    \item In Section \ref{section-statistic-results}, we provide complete experimental results, including the comparison of different variants on the whole benchmark and detailed comparison among DnBk, kDC-2 and -RR3.
\end{itemize}

\section{Heuristic Algorithm for Finding Initial Solution in DnBk}
\label{section-heuristic-algorithm}
The Algorithm \ref{alg_heuristic} shows the heuristic algorithm for building the initial solution. Note that the time complexity of the heuristic algorithm is $O(|V|+|E|)$.

\begin{algorithm}[ht!]
    \DontPrintSemicolon
        \caption{The Heuristic Algorithm to Produce an Initial $k$-defective Clique}
        \label{alg_heuristic}
        \KwIn{A graph $G=(V,E)$ and a non-negative integer $k$.}
        \KwOut{A heuristic solution $Q_h$ of maximum $k$-defective clique.} 
        \Begin{
            Initialize $Q_{h}\gets \emptyset$ \tcp{$Q_{h}$ maintains a $k$-defective clique.}
            Let $S\gets V(G)$\\
            \While{$G[S]$ is not a $k$-defective clique}{
                $v\gets$ the vertex with the smallest degree in $G[S]$\\
                $S\gets S\setminus v$
            }
            $Q_h\gets G[S]$\\
        }
        return $Q_{h}$
\end{algorithm}

\section {Proof to Lemma 2 and Lemma 3}
\label{section-lemma3-decompositioncorrect}

\begin{lemma}
Given an instance $I=(G, P,R)$ where $P\neq \emptyset$, the branching algorithm finds the maximum $k$-defective clique of $I$ in time $O(|R|^{2k}\gamma_c^{m})$ where 
$m=|CN(P)|$ and $\gamma_c$ is the base of the exponential factor in the time complexity for maximum clique algorithm. 
\end{lemma}
\begin{proof}
    It is known that the branching algorithm produces a depth-first tree.  
    Also, the running time of the algorithm is equal to the time to generate all the tree nodes.
    Regarding our branching algorithm, the tree nodes are partitioned into \textit{intermidate nodes} and \textit{leaf nodes}.  
\begin{itemize}
    \item For leaf nodes, we noticed that in these nodes, the set $P$ in the input instance $I=(G, P, R)$ satisfies $|P|\le 2k$. 
    Now, let $L(2k-|P|,|R|)$ be the upper bound of the number of leaf nodes.
    As when $|R|=0$ or $|P|=2k$, our algorithm invokes maximum clique computation, therefore:
    \[
    L(\cdot,0) = 1 \text{ and } L(0,\cdot) =1 
    \]
    By our binary branching rule, the linear recurrence for $L(2k-|P|,|R|)$ holds.
    \[
    L(2k-|P|,|R|) \geq  L(2k-|P|-1,|R|-1) + L(2k-|P|,|R|-1)
    \]
    So, any closed function of $L$ satisfying the above conditions is a valid upper bound for the number of leaf nodes. We claim that $L(2k-|P|,|R|) = (|R|+1)^{2k-|P|}$ satisfies the above conditions because
    \[
    \begin{aligned}
    & L(2k-|P|,|R|) \\
    & = (|R|+1)^{2k-|P|}  \\
    & = (|R|+1)(|R|+1)^{2k-|P|-1} \\
    & = |R|(|R|+1)^{2k-|P|-1} + (|R|+1)^{2k-|P|-1} \\
    & \geq |R|^{2k-|P|} + |R|^{2k-|P|-1} \\
    & = {(|R|-1+1)}^{2k-|P|} + {(|R|-1+1)}^{2k-|P|-1} \\
    & = L(2k-|P|,|R|-1) + L(2k-|P|-1,|R|-1)
    \end{aligned}
    \]
    Therefore, for problem instacne $I=(G, P, R)$, the number of leaf nodes is bounded by $ L(2k-|P|,|R|) = (|R|+1)^{2k-|P|} = O(|R|^{2k-|P|})$ asymptotically.

    Considering that the running time for each leaf node is $O(|R|^{O(1)}\gamma_c^{m})$ (which is the time of maximum clique computation), the whole running time for all leaf nodes is $O(|R|^{2k+O(1)}\gamma_c^{m})$ as $2k-|P| \leq 2k$.
    \item For intermediate nodes, we observed that our branching rule can produce at most two child nodes, and meanwhile, the running time for each node is polynomial.
    So the time of generating all intermediate nodes is also bounded by the $O(|R|^{2k-|P|})$, that is, the time of generating all child nodes.  So the running time is $O(|R|^{2k+O(1)})$.
\end{itemize}
    To sum up, the running time for $I=(G,P,R)$ is bounded by $O(|R|^{2k+O(1)}\gamma_c^{m})$ when $k$ is fixed.
    
\end{proof}

\begin{lemma}
Give an instance $I=(G,\emptyset,V)$ and an order of $V$ whic is denoted by $v_1,...,v_n$. 
Then $\omega_k(I)=\mathop{\max}_{i=[n]}\left(\omega_k(I'_i),\omega_k(I''_i)\right)$ where $I'_i=(G[\{v_i\} \cup N^+(v_i)\cup N^{2+}_{G}(v_i)],\{v_i\}, N^+(v_i)\cup N^{2+}_{G}(v_i))$ and $I''_i=(G[\{v_i\}\cup N^+(v_i)],\emptyset,N^+(v_i))$.   
\end{lemma}
\begin{proof}
    For any maximum $k$-defective clique $Q^*$ in $G$, suppose that $v_i\in Q^*$ is the first vertex in $Q^*$ with respect to the given order.    
    %there must exist an index $i\in [n]$ such that a maixmum $k$-defective clique $Q^*$ is in the subgraph induced by $\{v_i\}\cup N^+(v_i)\cup N^{2+}_{G}(v_i)$ and $Q^*$ must include vertex $v_i$. 
    We consider two cases that $v_i$ is either in $D(Q^*)$ or in $C(Q^*)$.
    \begin{itemize}
        \item If $v_i\in D(Q^*)$, then by Lemma 1(Two-hop Property) in the paper, any other vertex in $Q^*$ is a subset of $N^+(v_i)\cup N^{2+}_{G}(v_i)$. Then $\{v_i\}\subseteq D(Q^*)\subseteq \{v_i\}\cup N^+(v_i)\cup N^{2+}(v_i)$. So, $Q^*$ is the maximum $k$-defective clique to instance $I_i'$.
        \item If $v_i\in C(Q^*)$, then by the definition of $k$-defective set, any vertex in $D(Q^*)$ must be adjacent to $v_i$. That is to say, $\emptyset \subseteq D(Q^*) \subseteq N^+(v_i)$.
        So, $Q^*$ is the maximum $k$-defective clique to instance $I''_i$.
    \end{itemize}
\end{proof}

\section{Missing Proof for the Bound and Algorithms}
\label{section-bound-algorithm}

\subsection{Correctness of the Five Rules for Identifying Conflict Vertices}
We show that the following rules are correct for identifying conflict pairs.
\begin{enumerate}
    \item If $u\in R$, $v\in V\setminus (P\cup R)$ and $\{u,v\}\notin E$, then $\{u,v\}$ are   conflict vertices.
    \item If $u,v\in V\setminus (P\cup R)$ and $\{u,v\}\notin E$, then $\{u,v\}$ are conflict vertices.    
    \item If $u,v\in V\setminus P$ and $r(P\cup \{u,v\})<0$, then $\{u,v\}$ are conflict vertices.
    \item If $u,v\in V\setminus P$, $\{u,v\}\in E$ and $|N_G(u)\cap N_G(v)\cap V\setminus P|\le lb-(|P|+r(P)-w(u)-w(v)+2)$, then $\{u,v\}$ are conflict vertices.
    \item If $u,v\in V\setminus P$, $\{u,v\}\notin E$ and $|N_G(u)\cap N_G(v)\cap V\setminus P|\le lb-(|P|+r(P)-w(u)-w(v)+1)$, then $\{u,v\}$ are conflict vertices.
\end{enumerate}

\begin{proof}
The first two conditions hold because, in a $k$-defective clique, any vertex in the clique set should be adjacent to all other vertices.

The third condition is satisfied because $P\cup \{u,v\}$ is not a $k$-defective clique.  Due to the Hereditary Property, any $k$-defective clique subsumes $P$ cannot contain $u$ and $v$ at the same time.

Now, we show the correctness of the fourth and fifth rules. We denote $V\setminus (P\cup{\{u,v\}})$ as $T$ and $T$ can be partition into $T\cap N_G(u)\cap N_G(v)$ and $T\setminus{(N_G(u)\cap N_G(v))}$. For $P\cup \{u,v\}$, we can choose at most $r(P\cup \{u,v\})$ vertices from $T\setminus{(N_G(u)\cap N_G(v))}$ since every vertex in $T\setminus{(N_G(u)\cap N_G(v))}$ is not adjacent to at least one vertex in $\{u,v\}$. Therefore, the upper bound of $P\cup \{u,v\}$ is $|P\cup \{u,v\}|+|N_G(u)\cap N_G(v)\cap T|+r(P\cup \{u,v\})=|P|+2+|N_G(u)\cap N_G(v)\cap T|+r(P\cup \{u,v\})$. Therefore, if $\{u,v\}\notin E$, then $r(P\cup \{u,v\})=r(P)-w(u)-w(v)-1$; If $\{u,v\}\in E$, then $r(P\cup \{u,v\})=r(P)-w(u)-w(v)$. This further indicates that 
\begin{itemize}
    \item If $\{u,v\}\in E$ and $|P|+|N_G(u)\cap N_G(v)\cap T|+r(P)-w(u)-w(v)+2\leq lb$, then there cannot be a maximum $k$-defecitve clique include $P\cup \{u,v\}$ with size larger than $lb$, which means $\{u,v\}$ is a pair of conflict vertices. Then rule 4 is proved.
    \item If $\{u,v\}\notin E$ and $|P|+|N_G(u)\cap N_G(v)\cap T|+r(P)-w(u)-w(v)+1\leq lb$, then there cannot be a maximum $k$-defecitve clique include $P\cup \{u,v\}$ with size larger than $lb$, which means $\{u,v\}$ is a pair of conflict vertices. Then rule 5 is proved.
\end{itemize}
\textit{Note}: Rule 4 is also used in \cite{chang2023efficient} and \cite{gao2022exact} for edge deletion.
\end{proof}

\subsection{Proof of Bound Rule 3}
\label{subsection-bound}

\begin{bounding}[The Packing, Coloring with Conflict Bounding Rules]
\label{bounding_pack_color_exclusive}
    Given an instance $I=(G,P,R)$ and a lower bound $lb$, assume that $V\setminus P$ is partitioned into $\chi$ independent sets $\Pi_1,...,\Pi_\chi$. 
    For any $u,v\in V\setminus P$, $conflict(u,v)=1$ if $\{u,v\}$ are conflict vertices and $conflict(u,v)=0$ otherwise.    
        % Define $s_{i}=\sum_{j=1}^{|C_i|}x_{i_j}$.
    % If the optimal objective value of the optimization problem described below is not greater than $|Q'|$, then $\omega_k(I)=|Q'|$.    
     The optimal objective value of the optimization problem described in the following is an upper bound of $\omega_k(I)$.

\begin{align}
    \max_{S_i\subseteq \Pi_i,\forall i\in [\chi]}  & \quad |P|+\sum_{i=1}^{\chi}{|S_i|} & \text{ \textbf{OPT}} \nonumber\\
    \text{s.t.} &  \sum_{i=1}^{\chi}{\left( \binom{|S_i|}{2}+ \sum_{u\in S_i}w(u)\right)}\le r(P) \label{constraint_pack_color}\\
    & conflict(u,v)=0, \forall u,v \in S_i\ \forall i\in [\chi]   \label{constraint_conflict} 
\end{align}
\end{bounding}
\begin{proof}
We prove the optimal $k$-defective clique of instance $I$ satisfies all constraints and can provide a feasible solution of optimization problem \textbf{OPT}. 
Assume that $Q^*$ is an optimal $k$-defective clique and $|Q^*| > |P|$. Then $P\subseteq D(Q^*)\subseteq P\cup R$ and $Q^*\subseteq V\setminus P$. Let $S^*_i= Q^*\cap \Pi_i$ for any $i\in [\chi]$. Due to the definition of $k$-defective clique, we have 
\[
\sum_{i=1}^{\chi}{\binom{|S^{*}_i|}{2}}+ \sum_{i=1}^{\chi}{\sum_{u\in S^{*}_i}w(u)} + |E(\overline{G[Q]})| \le k
\]
where the first factor in the left-hand-side is the number of edges in each $\overline{G[S^*_i]}$, the second factor is the number of missing edges between each $S^*_i$ and $P$ and the third factor is the number of edges in $\overline{G[P]}$. 
Reorganizing the inequality, we have $\sum_{i=1}^{\chi}{\left( \binom{|S^*_i|}{2}+ \sum_{u\in S^*_i}w(u)\right)}\le r(P)$. This indicates that $S^*_i$ satisfies the constraint (2). By the definition of conflict vertices, every pair of vertices in $Q^*$ is not conflict, so for $\forall u,v\in S_i^*,\ \forall i\in [\chi]$ we have $conflict(u,v)=0$ then the constraint (3) is satisfied. Therefore, $|P|+\sum_{i=1}^{\chi}{|S^*_i|}=|Q^*|$ is a feasible solution to the \textbf{OPT}
So the optimal objective value of the Optimiation Problem \textbf{OPT} is an upper bound of $|Q^*|$.
\end{proof}
\subsection{Proof for Correctness of DPBound Algorithm}
\label{subsec_correct_packcolor}
We show that the DPBound algorithm solves the optimization problem OPT without the constraint (3), i.e., the conflict constraints.
\begin{proof}    
We show that in the first two steps, the algorithm actually computes the optimal value for 
$\max_{\forall i\in [\chi], S_i\subseteq \Pi_i}{|P|+\sum_{i=1}^{\chi}{|S_i|}}$ under the constraint (1) in the optimization problem.
The key point of the proof is that given an $i\le [\chi]$ and an $r\le r(P)$, the $t(i,r)$ that we computed in step 1, is the objective value of the following subproblem.
\begin{align}
    \max_{S_i\subseteq \Pi_i}  & \quad |S_i| \label{suboptimal_problem}\\
    \text{s.t.} &   \binom{|S_i|}{2}+ \sum_{u\in S_i}w(u) \le r \nonumber 
\end{align}
We show this fact by contraction.
Assume that there is another $S_i^{'}\subseteq \Pi_i$ such that $t'=\vert S'_i\vert> |S_i|=t(i,r)$ and $S_i^{'}$ satisfies $\sum_{u\in S_i^{'}}w(u)+\binom{|S_i^{'}|}{2}\leq r$. 
Assume that $S'_i$ is ordered as $u_i^1,...,u_i^{|S'_i|}$ by the order in step 1 of the DPBound algorithm (that is, the non-decreasing order of $w(\cdot)$). 
Now, it is safe to exchange $u_i^1$ with $v_i^1$ so that $S'_i$ still satisfies $\sum_{u\in S_i^{'}}w(u)+\binom{|S'_i|}{2}\leq r$. (Note that $v_i^1$ is the first vertex in $S_i$ by the DPBound algorithm.)
We continue this exchange procedure until all the first $|S_i|$ vertices in $S'_i$ are replaced by $v_i^1,...,v_i^{|S_i|}$. Because $|S'_i|>|S_i|$, there exist at least one vertex $u_i^{{|S_i|+1}}$ that $\binom{|S_i|+1}{2}+(\sum_{u\in S_i}w(u)+w(u_i^{{|S_i|+1}}))\le r$. This contradicts the fact that $|S_i|$ is the maximum number $j$ such that $\binom{j}{2}+\sum_1^{j}w(v_j^i)$ holds. Therefore, we claim that $t(i,r)$ is the optimal value for the optimization problem \ref{suboptimal_problem}.
%By the definition of $t(i,r)$, we can get $t(i,r)=|\Pi_i|$ or $\sum_{j=1}^{t(i,r)+1}w(v_j)+\binom{t(i,r)+1}{2}>r$.
%It is clear that $t(i,r)\le |\Pi_i)$ by the algorithm. 
%\begin{itemize}
%    \item If $t(i,r)=|\Pi_i|$ then we can get $\vert S'_i\vert>\vert \Pi_i\vert$, which is obviously a contradiction. 
%    \item If $t(i,r)<|\Pi_i|$, by the definition of $t(i,r)$, $\sum_{j=1}^{t(i,r)+1}w(v_j)+\binom{t(i,r)+1}{2}>r$. %Because vertices in $S_i\cup v_{|S_i|+1}$ are $|S_i|+1$ vertices with smallest weights. Hence we can get 
%$\sum_{u\in S'_i}w(u)+\binom{|S'_i|}{2}\geq \sum_{j=1}^{t(i,r)+1}w(v_j)+\binom{t(i,r)+1}{2}>r$, 
%which is contradicted to assumption.
%\end{itemize}

Based on the fact that $t(i,r)$ solves the sub-optimization problem for a given $i$ and $r$. It is simple to justify the $f(i,r)$ is the solution to Optimization Problem \ref{suboptproblem_pack_color} given $i\in [\chi]$ and $r\in \{0,...,r(P)\}$.
\begin{align}
    \max_{\forall j\in [i], S_j\subseteq \Pi_j}  & \quad \sum_{j=1}^{i}{|S_j|} \label{suboptproblem_pack_color}\\
    \text{s.t.} &  \sum_{j=1}^{i}{\left( \binom{|S_j|}{2}+ \sum_{u\in S_j}w(u)\right)}\le r \nonumber 
\end{align}

\begin{itemize}
    \item when $i=1$, it is clear that  $f(1, r)= t(1,r)$ for any $r\in \{0,...,r(P)\}$.
    \item When $i>1$,  $f(i,r)$ is the largest value over all $f(i-1,r')+t(i-1,r-r')$s for any $r'\in \{0,...,r\}$
\end{itemize}
Lastly, let $i=\chi$ and $r=r(P)$, we conclude that the DPBound algorithm solves the optimization problem in the packing-and-coloring bound.
\end{proof}

%\subsection{Proof for Correctness of PackColorConf Algorithm}
%We show that the PackColorConf algorithm get an upper bound of optimization problem \textbf{OPT}. %\textcolor{blue}{Obviously, our algorithm in the paper only consider part of the conflict pairs in %constraint (3), so the PackColorConf solved a relaxed version of \textbf{OPT}. Therefore, the %PackColorConf Algorithm provides an upper bound for the \textbf{OPT} hence can provide an upper bound for %$\omega(I)$.} 

\section{Proofs to the Dominance Relations Between Different Bounds}
\label{section-proof-dominance}
We first reformulate the sorting ~\cite{chang2023efficient} and club \cite{jin2024kd} bounding rules so that they match the context of our description.

\begin{bounding}[Sorting bound \cite{chang2023efficient}]
Given an instance $I=(G,P,R)$, assume that $V\setminus P$ is partitioned into $\chi$ independent sets $\Pi_1, \Pi_2, ..., \Pi_{\chi}$. Then, the following algorithm computes the upper bound for $\omega_k(I)$.
\begin{enumerate}
    \item For each $i\in [\chi]$, sort the vertices in each $\Pi_i$ in non-decreasing order regarding $w(\cdot)$. Assume the ordering as $v_i^1, v_i^2,...,v_i^{|\Pi_i|}$.  Then, for $j\in [|\Pi_i|]$, assign vertex $v_i^j$ the weight $w_s(v_i^j)=w(v_i^j)+j-1$.
    \item Sort all the vertices in $R$ by non-decreasing order of its weight $w_s(\cdot)$. Find the maximum $i_s$ such that $\sum_{j=1}^{i_s}w_{s}(v_j)\leq r(P)$
    \item return $|P|+i_s$ as the upper bound
\end{enumerate}    
\end{bounding}
%\paragraph{sorting bound}

\begin{bounding}[Club bound \cite{jin2024kd}]
    Given an instance $I=(G,P,R)$, the following algorithm computes the upper bound for $\omega_k(I)$.
\begin{enumerate}
    \item Sort the vertices in $V\setminus P$ by $w(\cdot)$ in non-decreasing order.  
    put vertex $u$ into $B_i$ if $w(u)=i$. Therefore, We have $R=B_0, B_1, ..., B_{r(P)}$ and note that $B_i=\emptyset$ if every vertex $v\in R$ such that $w(v)\neq i$.
    \item For each $i$ in $0,...,r_{r(P)}$, greedily partition vertices of $B_i$ into $\chi_i$ independent set.
    \item Re-partition $B_i$ into $z_i=\lceil \frac{|B_i|}{\chi_i}\rceil$ groups such that $B_i=b_1\cup b_2\cup ...\cup b_{z_i}$. We assume $B_i=\{v_i^1,v_i^2,...,v_i^{|B_i|}\}$, then for every $j\in [|B_i|]$, $v_i^j$ is belong to group  $b_{\lfloor \frac{j-1}{\chi_i} \rfloor+1}$. Assign the weight of vertex $u$ in group $j$ is $w_c(u)=w(u)+j-1$ (Note that $j\in \{1,2,...,z_i\}$).
    \item Sort all the vertices in $R$ with $w_c$ in increasing order. Then find the maximum $i_c$ such that $\sum_{j=1}^{i_c}w_{c}(v_j)\leq r(P)$
    \item return $|P|+i_c$ as the upper bound of $\omega_k(I)$
\end{enumerate}
\end{bounding}

%\paragraph{Club} 
%The step of Club
In the following, we slightly abuse the notation and use $A\ge B$ to represent that bounding rule A dominates bounding rule B.
Besides, we use Packing, Coloring, DPBound, and PackColorConf to denote the result obtained by the packing bound algorithm (see Bounding Rule 1 in the paper), coloring bound algorithm (see Bounding Rule 2 in the paper), DPBound algorithm (see Section 4.2.1 in the paper) and PackColorConf algorithms (see Section 4.2.2 in the paper), respectively. (Note that, the DPBound algorithm computes the optimal value of the optimization problem in Bounding Rule 3, but the PackColorConf algorithm computes only an upper bound of the optimization problem OPT.

\begin{theorem}
Given an instance $I=(G,P,R)$, assume that $V\setminus P$ is partitioned into $\chi$ independent sets $\Pi_1, \Pi_2, ..., \Pi_{\chi}$.
Then, the following dominance relations hold.
\begin{enumerate}
    \item  PackColorConf $\geq$ DPBound 
    \item DPBound $\ge $Sorting and  Sorting $\ge $ DPBound . 
    \item Sorting $\geq$ Packing
    \item DPBound $\geq$ Coloring
    \item Club $\geq$ Packing
\end{enumerate}    
\end{theorem}

\begin{proof}    
    \noindent\textbf{(1) PackColorConf$\geq$ DPBound } \\
    This relationship is quite straightforward. By the PackColorConf bounding rules, we consider conflict pairs, and hence $t(i,r)$ computed via this rule is smaller than or equal to $t(i,r)$ computed by the DPBound algorithm. (This has been demonstrated in the main body of the paper.) 
    On the other hand, the linear recurrence relations of computing $f(i,r)$ are the same in the two algorithms. Therefore, PackColorConf $\geq$ DPBound .\\
    
    \noindent\textbf{(2) Assume that the partition of $V\setminus P$ is the same for the DPBound and Sorting algorithm, the DPBound $\geq$ Sorting and  Sorting$\geq$DPBound .} \\
    Assume that the partition of $R$ is $\Pi_1,...,\Pi_\chi$. 
    Denote the upper bound obtained by DPBound is $ub_1$ and the upper bound obtained by Sorting is $ub_2$. 
    %We prove $ub_{1}=ub_{2}$, 
    
    \textit{(a) We first show that $ub_{2}\geq ub_{1}$, i.e., DPBound $\ge$ Sorting} \\
    As shown in the proof of the correctness of the DPBound algorithm (Section \ref{subsec_correct_packcolor}),  $t(i,r)$ is the maximum number of vertices that we can take from the set $\Pi_i$ into $P$ such that the number of missing edges is equal to or smaller than $r$.
    
    %We have proved that the packing-and-coloring algorithm can solve the optimization problem \ref{bounding_pack_color}. 
    Assume that a solution to the Optimization Problem OPT without constraint (2) is $S_1,S_2,...,S_{\chi}$.
    Let $r_1=\binom{|S_1|}{2}+\sum_{v\in S_1}w(v), r_2=\binom{|S_2|}{2}+\sum_{v\in S_2}w(v),...,r_\chi=\binom{|S_{\chi}|}{2}+\sum_{v\in S_{\chi}}w(v)$. 
    By our proof of the correctness of the DPBound algorithm (Section \ref{subsec_correct_packcolor}), we claim that  $f(\chi,r(P))$ obtained by DPBound in the last step is equal to $\sum_{i=1}^{\chi} t(i,r_i)$.
    Let us further denote $p_i=t(i,r_i)$. Then the following inequality holds for any $i\in [\chi]$
    \[
    \sum_{j=1}^{p_i}w(v_j^i)+\binom{p_i}{2}=\sum_{j=1}^{p_i}(w(v_j^i)+j-1)=\sum_{j=1}^{p_i}w_s(v_j^i)= r_i
    \]
    The first equality holds due to the computation steps of $t(i,r_i)$ and the second equality holds due to the definition of $w_s(v_j^i)$.
    Now, considering $S_1,S_2,...,S_{\chi}$ is a feasible solution to the Optimization Problem in the packing-and-coloring bound, we are ready to obtain the following inequalities. 
    \[
    \sum_{i=1}^{f(\chi,r(P))}w_s(v_i)=\sum_{i=1}^{\chi}\sum_{j=1}^{p_i}w_s(v_j^i)= \sum_{i=1}^{\chi} r_i\leq r(P)
    \]
    By the first two steps in the Sorting bound, we finally obtained $ub_2-|P|\ge f(\chi, r(P))$, that is, $ub_1\leq ub_2$

   \textit{(b) Then, we show that $ub_{1}\geq ub_{2}$, i.e., Sorting $\geq$ DPBound } \\
    By the steps of computing the sorting bound, $\sum_{j=1}^{i_s}w_{s}(v_j)\leq r(P)$.
    Let us denote sequence $R_s=\{v_1,v_2,...,v_{i_s}\}$ are the first $i_s$ vertices in $V\setminus P$, sorting by the non-decreasing order of $w_s(\cdot)$. Note that we slightly abuse the notation by using $R_s$ to also denote the set of $R_s$ when the context is clear.
    For each $i\in [\chi]$, let us assume that sequence $T_i$ is a subsequence of $R_s$ and  $T_i = R_s \cap \Pi_i=\{ v_1^i,v_2^i,...,v_{|T_i|}^i\}$.
    %where $n_i=|T_i|$. 
    Now, we get the following equality due to the definitions of $w_s(\cdot)$.

    \[
        \begin{aligned}
            &\sum_{j=1}^{i_s}w_{s}(v_j)\\
            &=\sum_{i=1}^{\chi}\sum_{j=1}^{|T_i|}w_s(v_j^i)\\
            &=\sum_{i=1}^{\chi}\sum_{j=1}^{|T_i|}(w(v_j^i)+j-1)\\
            &=\sum_{i=1}^{\chi}\left(\binom{|T_i|}{2}+\sum_{j=1}^{|T_i|}w(v_j^i)\right)
        \end{aligned}
    \]

    Recall the fact that $\sum_{j=1}^{i_s}w_{s}(v_j)\leq r(P)$, we have
    \[
    \sum_{i=1}^{\chi}\left(\binom{|T_i|}{2}+\sum_{j=1}^{|T_i|}w(v_j^i)\right)\leq r(P)
    \]
    %Therefore, $T_i$ satisfies the constraint in the optimization problem in the packing-and-coloring bound.
    Therefore, the partition $T_1,...,T_\chi$ is a feasible solution to the Optimization Problem in the packing-and-coloring bound.
    On the other hand, we already demonstrate that $ub_1$ is the optimal solution to the packing-and-coloring bound. 
    Thus, we have $ub_1\geq ub_2$.
    
    %which means $T_i=\{v_1,v_2,...,v_{i_s}\}\cap \Pi_i$, and we denote $|T_i|$ as $n_i$. The vertices in $T_i$ with non-%decreasing $w(\cdot)$ order is $T_i=\{v_1^i,v_2^i,...,v_{T_i}^i\}$. \\
    %$\because$ $\sum_{j=1}^{i_s}w_{s}(v_j)\leq r(P)$ and $\sum_{j=1}^{i_s}w_{s}
    %(v_j)=\sum_{i=1}^{\chi}\sum_{j=1}^{|T_i|}w_s(v_j^i)=\sum_{i=1}^{\chi}\sum_{j=1}^{|T_i|}(w(v_j^i)+j-%1)=\sum_{i=1}^{\chi}(\binom{|T_i|}{2}+\sum_{j=1}^{|T_i|}w(v_j^i))$\\
    %$\therefore$ $\sum_{i=1}^{\chi}(\binom{|T_i|}{2}+\sum_{j=1}^{|T_i|}w(v_j^i))\leq r(P)$\\
    %$\therefore$ $T_1,T_2,...,T_{\chi}$ is a feasible solution of optimization problem (\ref{bounding_pack_color})\\
    %$\therefore$ By the proof of section 2.2 it is clear that $|P|+\sum_{i=1}^{\chi}|T_i|\leq |P|+f(\chi, r(P))$, %which means $ub_1\ge ub_2$.\\
    \noindent\textbf{(3) Sorting$\geq$ Packing} \\
    Recall the steps of computing the packing bound: Sort the vertex in $R$ by non-decreasing order of the weight $w(\cdot)$ and find the largest index $i_p$ such that $\sum_{j=1}^{i_p}w(v_j)\leq r(P)$.
    The sorting bound has a similar step, but the definition of weight is different from the packing bound.
    In sorting bound, given a vertex $u\in R$, suppose $u\in \Pi_i$ and $u$ ranks at the $j$th position by non-decreasing order of $w(\cdot)$ in $\Pi_i$, then the weight of $u$ defined as $w_{s}(u)=w(u)+j-1$.
    %Whereas the weight for $\forall u\in R$ assigned by packing bound is $w_{p}(u)=w(u)$. 
    Clearly,  for each $u\in R$, we have $w(u)\leq w_{s}(u)$. 
    %Now, let $i_s$ be the maximum number such that $\sum_{j=1}^{i_s}w_s(v_j)\leq r(P)$ , and $i_p$ be the maximum number such that $\sum_{j=1}^{i_p}w(v_j)\leq r(P)$ 
    Now suppose that $R$ is ordered by $w_s(\cdot)$.  Then $\sum_{j=1}^{i_s}w(u_j)\leq r(P)$  holds. 
    Therefore, $i_p$ is larger than or equal to $i_s$. 
    Hence, Sorting$\geq$ Packing.

    \noindent\textbf{(4)DPBound $\geq$ Coloring}\\
    Suppose that $S_1,...,S_\chi$ is an optimal solution to  Optimization Problem (1) in packing-and-coloring bound. %Then the result obtained by the DPBound algorithm is $|P|+\sum_i^{\chi}|S_i|$.
    Then $S_i\subseteq \Pi_i$, we have $|S_i|\le |\Pi_i|$. 
    %Assume that $S_i=\{v_i^1, v_i^2,...,v_i^{|S_i|}\}$.
    %Then,
    On the other hand,$\binom{|S_i|}{2}\le k$ because
    \[
     \begin{aligned}
         &\binom{|S_i|}{2}\\
     &\le \binom{|S_i|}{2} + \sum_{u\in S_i}w(u) \\
     &\leq \sum_{i=1}^{\chi}\left(\binom{|S_i|}{2} + \sum_{u\in S_i}w(u)\right)\leq r(P) \le k.
     \end{aligned}
    \]
    In sum, we have $|S_i|\leq \min \{\lfloor\frac{1+\sqrt{1+8k}}{2}\rfloor,|\Pi_i|\}$. 
    This indicates $|P|+\sum_{i=1}^{\chi}|S_i|\le |P|+\sum_{i=1}^{\chi}{\min \{\lfloor\frac{1+\sqrt{1+8k}}{2}\rfloor,|\Pi_i|\}}$, that is, DPBound $\geq$ Coloring. \\
    \noindent\textbf{(5)Club$\geq$ Packing}\\
    If suffices to justify that for any vertex $u\in R$, the weight defined by club bounding, which is denoted as $w_c(u)$, is larger than or equal to $w(u)=w_{p}(u)$. Then the club rule dominates the packing bounding rule ( just like our proof to (3) Sorting $\geq$ Packing ).
   
\end{proof}

\section{More Statistical Results of Experiments}
\label{section-statistic-results}
\subsection{Complete Experiments for Bounds}

\begin{figure*}[ht]
\vspace{-2mm}
\includegraphics[width=\linewidth]{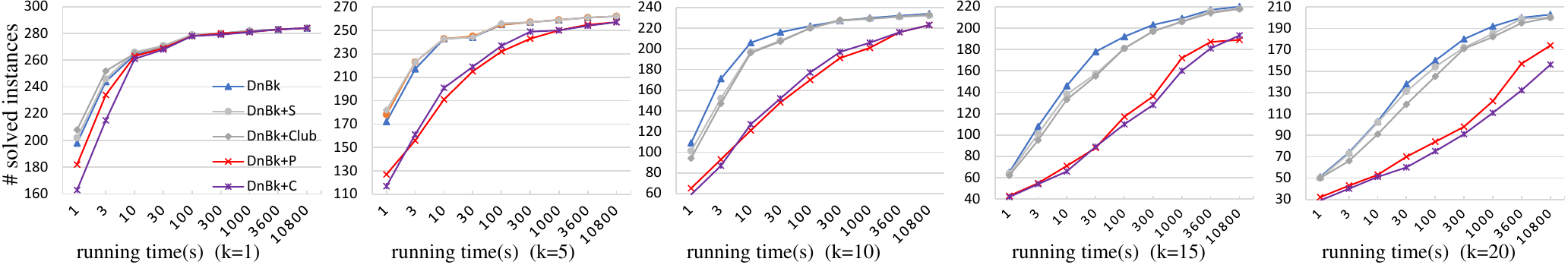}
\caption{
Number of Solved Instances for 3 Datasets for DnBk, DnBk\textsubscript{S}, DnBk\textsubscript{Club}, DnBk\textsubscript{P}, DnBk\textsubscript{C} when $k=1,5,10,15,20$.}
\label{fig:boundCmp}
% \vspace{-2mm}
\end{figure*}
Detailed experimental results are given in file data.pdf. Note that for certain $k$, if the maximum $k$-defective clique is smaller than $k+2$, then the \textit{opt} is assigned to $k+1$. 
Here we will give the whole performance of the five variants.
We first show the results of the five variants in Fig. \ref{fig:boundCmp} on all graphs with $k=1,5,15,20$. From Fig. \ref{fig:boundCmp}, we can see that for all $k$ values DnBk, DnBk\textsubscript{S} and DnBk\textsubscript{Club} are faster than DnBk\textsubscript{P} and DnBk\textsubscript{C}. Moreover, the dominance is become more obvious when $k$ is increasing. When $k$ is relatively small, i.e., $k=1,5$, the efficiency of DnBk, DnBk\textsubscript{S} and DnBk\textsubscript{Club} are almost the same. But when $k$ is large, i.e., $k=15,20$, DnBk outperforms DnBk\textsubscript{S} and DnBk\textsubscript{Club} and when $k=20$ DnBk\textsubscript{S} obviously outperforms DnBk\textsubscript{Club}. For DnBk\textsubscript{P} and DnBk\textsubscript{C}, they cannot dominant each other but DnBk\textsubscript{P} is slightly better than DnBk\textsubscript{C}.

\subsection{Comparing to kDC-2 and -RR3}
Since the DnBk performs very similar to -RR3 and kDC-2 in \cite{chang2024maximum}, we show the experiment results of these three solvers on the 39 Socfb large social network graphs. The three solvers are listed below.
\begin{itemize}
    \item DnBk, the most efficient solver in our paper.
    \item kDC-2, the most efficient solver in \cite{chang2024maximum}
    \item -RR3, the variant of kDC-2 without \textit{Reduction Rule 3} in ~\cite{chang2024maximum}.
\end{itemize}

Actually, the research [9] is concurrent with ours, and the source code is still unavailable.
Hence, we can only compare DnBk with kDC-two and -RR3 by analyzing the Table3 in [9], i.e. running time on 39 socfb-graphs with k=10 and 15, the only detailed results reported.

First, we use kDC as the benchmark algorithm for machine performance since both studies used the same kDC source code. We observed that the machine in [9] is likely to be 2.21x faster than ours.
Thus, we scaled the running times in [9] by 2.21 times and compare them with DnBk. 

The processed running times of all instances solved by kDC-2 and -RR3 are listed in columns "KDC-2" and "-RR3"  of Table 16 in Section 2 of the data.pdf respectively. The data.pdf can be obtained from \url{https://github.com/cy-Luo000/Maximum-k-Defective-Clique.git}.

From Table 16, we can get the sum of running times of 39 Socfb graphs for each solver and different $k$ values respectively, which is shown in Table \ref{tab:total-running-time}.

% Table generated by Excel2LaTeX from sheet 'Sheet2'
\begin{table}[htbp]
\caption{The total running time for DnBk, kDC-2 and -RR3 when $k=10$ and $k=15$. The time unit is second. The best performers are highlighted in bold.}
  \centering
  
    \begin{tabular}{c|ccc}
     \toprule
          & DnBk  & kDC-2 & -RR3 \\
    \hline
    k=10  & 380.92  & \textbf{309.45 } & 556.93  \\
    k=15  & \textbf{3206.40 } & 5966.90  & 13728.54  \\
    \bottomrule
    \end{tabular}%
  \label{tab:total-running-time}%
\end{table}%

Generally, when k = 10, we have kDC-two$>$DnBk$>$RR3 (here “$>$” means “faster than”). But when k=15, DnBk $>$ kDC-two$ >$RR3.

Specifically, for instance which can be solved by all algorithms within 10 seconds, kDC-two is generally faster than DnBk. For instances that cannot be solved by kDC-two or RR3 in 100 seconds, DnBk is faster than them in most cases. For reference, we list the detailed running times of the second group of instances, given k=15, which are listed in Table \ref{tab:hard-exp}.

\begin{table}[H]
\caption{The running times of all the instances that kDC-2 and -RR3 are unable to solve in 100 seconds when $k=15$. The time unit is second and the best performers are highlighted in bold.}
\vspace{\topsep}
  \centering
    \begin{tabular}{cccc}
    \toprule
    Graph & DnBk  & kDC-2 & -RR3 \\
    \midrule
    socfb-A-anon & 685.71  & \textbf{528.62 } & 1386.79  \\
    socfb-Auburn71 & \textbf{51.57 } & 117.22  & 327.34  \\
    socfb-B-anon & \textbf{1785.76 } & 3200.45  & 6148.75  \\
    socfb-Oklahoma97 & \textbf{31.70 } & 110.59  & 168.10  \\
    socfb-Texas84 & \textbf{36.95 } & 188.00  & 805.09  \\
    socfb-UGA50 & \textbf{382.97 } & 1484.10  & 4144.88  \\
    \bottomrule
    \end{tabular}%
  \label{tab:hard-exp}%
\end{table}

%%%%%%%%%%%%%%%%%%%%%%%%%%%%%%%%%%%%%%%%%%%%%%%%%%%%%%%%%%%%%%%%%%%%%%%%

%%% Use this command to include your bibliography file.

\bibliography{Supplementary}